\documentclass[11pt]{article}
\pdfoutput=1
\usepackage[protrusion=true,expansion=true]{microtype}
\usepackage{amsmath,amssymb,amsthm}
\usepackage{graphicx}
\usepackage{fullpage}
\usepackage[backref=page]{hyperref}
\usepackage{color}
\usepackage{wrapfig}
\usepackage{tikz}
\usetikzlibrary{decorations.pathreplacing}
\usepackage{setspace}
\usepackage{algorithm}
\usepackage[noend]{algpseudocode}
\usepackage[framemethod=tikz]{mdframed}
\usepackage{xspace}
\usepackage{pgfplots}
\usepackage{framed}
\usepackage{subcaption}
\usepackage{thmtools}
\usepackage{thm-restate}
\pgfplotsset{compat=1.5}

\newtheorem{theorem}{Theorem}[section]

\newtheorem{lemma}[theorem]{Lemma}

\newtheorem{definition}[theorem]{Definition}

\newtheorem{fact}[theorem]{Fact}

\newenvironment{proofof}[1]{\begin{trivlist} \item {\bf Proof
#1:~~}}
  {\qed\end{trivlist}}

\newcommand{\namedref}[2]{\hyperref[#2]{#1~\ref*{#2}}}
\newcommand{\thmlab}[1]{\label{thm:#1}}
\newcommand{\thmref}[1]{\namedref{Theorem}{thm:#1}}
\newcommand{\lemlab}[1]{\label{lem:#1}}
\newcommand{\lemref}[1]{\namedref{Lemma}{lem:#1}}

\newcommand{\seclab}[1]{\label{sec:#1}}
\newcommand{\secref}[1]{\namedref{Section}{sec:#1}}
\newcommand{\applab}[1]{\label{app:#1}}
\newcommand{\appref}[1]{\namedref{Appendix}{app:#1}}
\newcommand{\factlab}[1]{\label{fact:#1}}
\newcommand{\factref}[1]{\namedref{Fact}{fact:#1}}

\newcommand{\figlab}[1]{\label{fig:#1}}
\newcommand{\figref}[1]{\namedref{Figure}{fig:#1}}

\newcommand{\deflab}[1]{\label{def:#1}}
\newcommand{\defref}[1]{\namedref{Definition}{def:#1}}

\def \View    {\mdef{\mathsf{View}}}
\def \Ber    {\mdef{\mathsf{Ber}}}
\def \Polya    {\mdef{\mathsf{Polya}}}

\def \DLap    {\mdef{\mathsf{DLap}}}

\newcommand{\PPr}[1]{\ensuremath{\mathbf{Pr}\left[#1\right]}}

\newcommand{\Ex}[1]{\ensuremath{\mathbb{E}\left[#1\right]}}
\newcommand{\EEx}[2]{\ensuremath{\mathbb{E}_{#1}\left[#2\right]}}
\renewcommand{\O}[1]{\ensuremath{\mathcal{O}\left(#1\right)}}

\newcommand{\eps}{\varepsilon}

\def \calA    {\mdef{\mathcal{A}}}

\def \calD    {\mdef{\mathcal{D}}}
\def \calE    {\mdef{\mathcal{E}}}

\def \calM    {\mdef{\mathcal{M}}}

\def \calP    {\mdef{\mathcal{P}}}

\def \calR    {\mdef{\mathcal{R}}}
\def \calS    {\mdef{\mathcal{S}}}
\def \calT    {\mdef{\mathcal{T}}}
\def \calU    {\mdef{\mathcal{U}}}
\def \calV    {\mdef{\mathcal{V}}}

\def \calX    {\mdef{\mathcal{X}}}
\def \sfV    {\mdef{\mathsf{V}}}
\def \sfW    {\mdef{\mathsf{W}}}
\def \sfX    {\mdef{\mathsf{X}}}
\def \sfU    {\mdef{\mathsf{U}}}

\def \calZ    {\mdef{\mathcal{Z}}}

\def \bbZ   {\mdef{\mathbb{Z}}}

\newcommand{\mdef}[1]{{\ensuremath{#1}}\xspace}  

\DeclareMathOperator*{\TVD}{TVD}

\DeclareMathOperator*{\Swap}{Swap}

\newcommand{\ceil}[1]{\mdef{\left\lceil#1\right\rceil}}               

\newcommand{\ignore}[1]{}

\newif\ifnotes\notestrue 
\ifnotes
\else
\fi

\hypersetup{
     colorlinks   = true,
     citecolor    = mahogany,
	 linkcolor	  = bleudefrance,
	 urlcolor     = mahogany
}

\makeatletter
\renewcommand*{\@fnsymbol}[1]{\textcolor{mahogany}{\ensuremath{\ifcase#1\or *\or \dagger\or \ddagger\or
 \mathsection\or \triangledown\or \mathparagraph\or \|\or **\or \dagger\dagger
   \or \ddagger\ddagger \else\@ctrerr\fi}}}
\makeatother

\providecommand{\email}[1]{\href{mailto:#1}{\nolinkurl{#1}\xspace}}

\definecolor{mahogany}{rgb}{0.75, 0.25, 0.0}
\definecolor{bleudefrance}{rgb}{0.19, 0.55, 0.91}
\definecolor{darkgoldenrod}{rgb}{0.72, 0.53, 0.04}

\definecolor{Gred}{RGB}{219, 50, 54}
\definecolor{Ggreen}{RGB}{60, 186, 84}
\definecolor{Gblue}{RGB}{72, 133, 237}
\definecolor{Gyellow}{RGB}{247, 178, 16}
\definecolor{ToCgreen}{RGB}{0, 128, 0}
\definecolor{myGold}{RGB}{231,141,20}
\definecolor{myBlue}{rgb}{0.19,0.41,.65}
\definecolor{myPurple}{RGB}{175,0,124}

\usepackage[colorinlistoftodos,prependcaption,textsize=scriptsize]{todonotes}
\providecommand{\Comments}{0}  
\ifnum\Comments=1
\paperwidth=\dimexpr \paperwidth + 4.4cm\relax
\oddsidemargin=\dimexpr\oddsidemargin + 2.2cm\relax
\evensidemargin=\dimexpr\evensidemargin + 2.2cm\relax
\fi
\newcommand{\mytodo}[1]{\ifnum\Comments=1{#1}\fi}

\newcommand{\pasin}[1]{\mytodo{\todo[linecolor=Gred,backgroundcolor=Gred!25,bordercolor=Gred]{Pasin: #1}}}
\newcommand{\samson}[1]{\mytodo{\todo[linecolor=darkgoldenrod,backgroundcolor=darkgoldenrod!25,bordercolor=darkgoldenrod]{Samson: #1}}}

\begin{document}
\allowdisplaybreaks

\title{Differentially Private Aggregation via Imperfect Shuffling}
\author{Badih Ghazi\thanks{Google Research. E-mail: \email{badihghazi@gmail.com}.}
\and
Ravi Kumar\thanks{Google Research. E-mail: \email{ravi.k53@gmail.com}.}
\and
Pasin Manurangsi\thanks{Google Research. E-mail: \email{pasin@google.com}.}
\and
Jelani Nelson\thanks{UC Berkeley and Google Research. E-mail: \email{minilek@alum.mit.edu}.}
\and
Samson Zhou\thanks{Texas A\&M University. Work done while at UC Berkeley and Rice University. E-mail: \email{samsonzhou@gmail.com}.}
}
\date{\today}

\maketitle
\begin{abstract}
In this paper, we introduce the imperfect shuffle differential privacy model, where messages sent from users are shuffled in an \emph{almost} uniform manner before being observed by a curator for private aggregation. We then consider the private summation problem. We show that the standard split-and-mix protocol by Ishai et. al. [FOCS 2006] can be adapted to achieve near-optimal utility bounds in the imperfect shuffle model. Specifically, we show that surprisingly, there is no additional error overhead necessary in the imperfect shuffle model. 
\end{abstract}

\section{Introduction}
Differential privacy (DP)~\cite{DworkMNS06} has emerged as a popular concept that mathematically quantifies the privacy of statistics-releasing mechanisms. 
Consequently, DP mechanisms have been recently deployed in industry~\cite{Greenberg16,ErlingssonPK14,shankland2014google,DingKY17}, as well as by government agencies such as the US Census Bureau~\cite{abowd2018us}. 
DP is parameterized by $\eps$ and $\delta$, where $\eps$ is a privacy loss parameter that is generally a small positive constant such as $1$ and $\delta$ is an approximation parameter or ``failure'' probability that is typically (smaller than) inverse-polynomial in $n$:
\begin{definition}[Differential privacy]
\deflab{def:dp}
\cite{DworkMNS06, dwork2006our}
Given $\eps>0$ and $\delta\in(0,1)$, a randomized algorithm $\calA:X\to Y$ is $(\eps,\delta)$-differentially private if, for every neighboring datasets $x$ and $x'$, and for all $S\subseteq Y$,
\[\PPr{\calA(x)\in S}\le e^{\eps}\cdot\PPr{\calA(x')\in S}+\delta.\]
\end{definition}

In this paper, we study the real summation problem, where each of $n$ parties holds a number $x_i\in[0,1]$ for all $i\in[n]$ and the goal is to privately (approximately) compute $\sum_{i=1}^n x_i$. 
Due to its fundamental nature, the private real summation problem has a wide range of applications, such as private distributed mean estimation~\cite{SureshYKM17,BiswasD0U20}, e.g., in federated learning~\cite{KonecnyMYRSB16,GirgisDDKS21,KairouzMABBBBCC21}, private stochastic gradient descent~\cite{SongCS13,BassilyST14,AbadiCGMMT016,AgarwalSYKM18,ChenWH20}, databases and information systems~\cite{KotsogiannisTHF19,WilsonZLDSG20}, and clustering~\cite{StemmerK18,Stemmer21}. 

In the central model of DP, where a curator is given full access to the raw data in order to release the private statistic or data structure,  the Laplace mechanism~\cite{DworkMNS06} can achieve, for real summation, additive error $\O{\frac{1}{\eps}}$, which is known to be nearly optimal for $\eps\le 1$~\cite{GhoshRS12}. 

However, the ability for the curator to observe the full data is undesirable in many commercial settings, where the users do not want their raw data to be sent to a central curator. 
To address this shortcoming, the local model of DP~\cite{KasiviswanathanLNRS11,warner1965randomized} (LDP) demands that all messages sent from an individual user to the curator is private. 
Unfortunately, although the local model achieves near-minimal trust assumptions, numerous basic tasks provably must suffer from significantly larger estimation errors compared to their counterparts in the central model. 
For the real summation problem, \cite{BeimelNO08} achieves additive error $\mathcal{O}_{\eps}(\sqrt{n})$ and it is known that smaller error bounds cannot be achieved~\cite{ChanSS12}. 

Consequently, the shuffle model~\cite{BittauEMMRLRKTS17,ErlingssonFMRTT19,CheuSUZZ19} of DP was introduced as an intermediary between the generous central model and the strict local model. 
In the shuffle model, the messages sent from the users are randomly permuted before being observed by the curator, in an encode-shuffle-analyze architecture. 
Surprisingly, when users are allowed to send multiple messages, there exist protocols in the shuffle model of DP that achieve additive error $\mathcal{O}_{\eps}(1)$ for the private real summation problem~\cite{GhaziMPV20,BalleBGN20,GhaziKMPS21}. 
Unfortunately, practical applications can lack the ideal settings that provide the full assumptions required by the shuffle model of DP. 

\subsection{Model and Motivation}
\seclab{sec:model-mot}
We first define a natural generalization of the uniform shuffler that tolerates imperfections. 
Let $\Pi$ be the set of permutations on $[n]$.  
For $\pi,\pi'\in\Pi$, we define $\Swap(\pi,\pi')$ to be the minimum number of coordinate \emph{swaps}\footnote{We say that $\pi'$ results from an application of a coordinate swap on $\pi$ if and only if $\pi(i) = \pi'(i)$ on all except two $i \in [n]$.} that can be applied to $\pi$ to obtain $\pi'$. 

\begin{definition}[$\gamma$-Imperfect Shuffler]
\deflab{def:imp:shuffle}
For a distortion parameter $\gamma>0$, we say that $\calS$ is a $\gamma$-imperfect shuffler if, for all $\pi,\pi'\in\Pi$,
\[\PPr{\calS=\pi}\le e^{\gamma \cdot \Swap(\pi,\pi')}\PPr{\calS=\pi'}.\]
\end{definition}

We call an output from such a shuffler a \emph{$\gamma$-imperfect shuffle} or a \emph{$\gamma$-I-shuffle}, for short. 
Here, $\gamma$ represents an upper bound on the multiplicative distortion of the output probabilities of the distributions of the shuffler, i.e., how the distribution deviates from a perfectly symmetric shuffler. 
For example, $\gamma=0$ corresponds to a perfectly symmetric shuffler while $\gamma\to\infty$ represents almost no guarantee from the shuffler. 

To understand the motivation behind this definition, consider a setting where a number of user devices collect statistics to be sent to an intermediate buffer, which is ultimately sent to a central curator for processing. 
The devices may choose to perform this collection over different periods of time, so that immediately sending their statistics over to the curator could reveal information about their identity, through the timestamp.  

For example, consider a setting where sensors are monitoring traffic in US cities during peak afternoon hours. 
Then reports that are received earlier in the day by the curator are more likely to correspond to cities that are in the east, while reports that are received later in the day by the curator are more likely to correspond to cities in the west.  
To mitigate this, the sensors instead could choose a universally fixed hour during the day to broadcast their reports from the previous day, at some random time during the hour.

Specifically, each user $i\in[n]$ could choose a time $t_i$, say normalized without loss of generality to $t_i\in[0,1]$, and send their messages at time $t_i$. 
If the $t_i$ are chosen uniformly at random and this protocol was executed perfectly, it would result in a uniform shuffle of the messages for a buffer that strips both the source information and the exact time of arrival, e.g., \cite{ThomsonW23}.\footnote{We assume in this example that the buffer can queue the messages, and then forward them to the analyst at some point of time, but that it cannot further shuffle them. 
The (imperfect) shuffling we consider stems solely from the randomization of the transmission time of the messages by the users.}
However, issues may arise such as different clock skews, where users may not perfectly synchronize the fixed hour during which the messages should be sent, or communication delays, either because an intermediate link has failed or simply because latency varies across different networks. 
It is unclear how to model the imperfect shuffle resulting from these issues using the standard shuffle model. 

For a better handle on modeling the imperfection, we can assume that each $t_i$ is adversarially chosen in $[0,1]$. 
Moreover, each message transmission time can now be altered by a random offset from the intended release time, where the offset is drawn, e.g., from a Laplacian distribution. 
Specifically, each user $i\in[n]$ draws an offset $\tau_i$ from the (centered) Laplacian distribution with scale $\frac{2}{\gamma}$, and sends their message at time $t_i+\tau_i$ instead of at time $t_i$. 

In other words, each user $i\in[n]$ sends their message at time $t_i+\tau_i$, which is determined by the two following quantities:
\begin{enumerate}
\item 
$t_i$ is an arbitrary and possibly adversarially chosen offset due to nature or some other external source, e.g., clock skews, transmission failure, communication delay.
\item 
$\tau_i$ is an internal source of noise that the protocol can sample from a predetermined distribution to mitigate the negative privacy effects of $t_i$.
\end{enumerate}

Note that whereas two permutations $\pi, \pi'$ on $[n]$ with swap distance one were equally likely to be output by the shuffler, this may now no longer be the case. 
On the other hand, for fixed $i,j\in[n]$ and conditioning on the values of $\{t_1,\tau_1,\ldots,t_n,\tau_n\}\smallsetminus\{t_i,\tau_i,t_j,\tau_j\}$, we can see that for $a,b\in[n]$, the probability that $t_a+\tau_a\le t_i+\tau_i\le t_{a+1}+\tau_{a+1}$ and $t_b+\tau_b\le t_j+\tau_j\le t_{b+1}+\tau_{b+1}$ is within an $e^{\gamma}$ factor of the probability that $t_a+\tau_a\le t_j+\tau_j\le t_{a+1}+\tau_{a+1}$ and $t_b+\tau_b\le t_i+\tau_i\le t_{b+1}+\tau_{b+1}$. 


Specifically, let $\calE_1$ be the event that $\tau_i\in[t_a+\tau_a-t_i ,t_{a+1}+\tau_{a+1}-t_i]$, where $\tau_i$ is a (centered) Laplace random variable and scale $\frac{2}{\gamma}$. 
Similarly, let $\calE_2$ be the event that $\tau_j\in[t_b+\tau_b-t_j, t_{b+1}+\tau_{b+1}-t_j]$ where $\tau_j$ is a (centered) Laplace random variable and scale $\frac{2}{\gamma}$. 
Furthermore, let $\calE_3$ be the event that $\tau_j\in[t_a+\tau_a-t_j ,t_{a+1}+\tau_{a+1}-t_j]$ and $\calE_4$ be the event that $\tau_i\in[t_b+\tau_b-t_i, t_{b+1}+\tau_{b+1}-t_i]$. 
Then by the properties of the Laplace distribution and the assumption that $t_i, t_j \in [0, 1]$, we have 
$\PPr{\calE_1\wedge\calE_2} = \PPr{\calE_1}\PPr{\calE_2}\le (e^{\gamma/2} \cdot \PPr{\calE_3}) (e^{\gamma/2} \cdot \PPr{\calE_4}) = e^\gamma \cdot \PPr{\calE_3\wedge\calE_4}.$ 
Thus, the resulting distribution over permutations is captured by the $\gamma$-I-shuffle model. 

We can naturally generalize this setting to the model where each user sends $m$ messages, e.g., $m$ buffers collect messages from $n$ users, which results in times $\{t_{i,j}\}_{i\in[n],j\in[m]}$ and offsets $\{\tau_{i,j}\}_{i\in[n],j\in[m]}$. 
Formally, for $m$ rounds of messages for the $n$ users, $\{m_{i,j}\}_{i\in[n],j\in[m]}$, a separate permutation $\pi_j$ drawn from a $\gamma$-imperfect shuffler is used to shuffle the messages $\{m_{i,j}\}_{i\in[n]}$, for each $j\in[m]$. 
For example, $\{m_{i,1}\}_{i\in[n]}$ is shuffled according to a permutation $\pi_1$ drawn from a $\gamma$-imperfect shuffler, $\{m_{i,2}\}_{i\in[n]}$ is shuffled according to an independent permutation $\pi_2$ drawn from the same $\gamma$-imperfect shuffler, and so on and so forth. 

We remark that the above model is sometimes referred to as the \emph{$m$-parallel shuffling} model; another model that has been considered in literature is one where all the $mn$ messages are shuffled together using a single shuffler. 
We only focus on the former in this paper.
It remains an interesting open question whether our results can be extended to the latter model.

\subsection{Our Contributions}
Surprisingly, we present a protocol for the real summation problem that matches the utility bounds of the best protocols in the shuffle model. 
Thus, we show that there is no additional error overhead necessary in the $\gamma$-I-shuffle model, i.e., there is no utility loss due to the imperfect shuffler. 

\begin{restatable}{theorem}{thmmain}
\thmlab{thm:main}
Let $n\ge 19$ and $\gamma\le\frac{\log\log n}{80}$ be a distortion parameter. 
Then there exists an $(\eps,\delta)$-DP protocol for summation in the $\gamma$-I-shuffle model with expected absolute error $\O{\frac{1}{\eps}}$ and 
$m=\O{e^{4\gamma}+\frac{e^{4\gamma}(\log\frac{1}{\delta}+\log n)}{\log n}}$
messages per party. 
Each message uses $\O{\log q}$ bits, for $q=\ceil{2n^{3/2}}$. 
\end{restatable}

Observe that when $\delta$ is inverse-polynomial in $n$ and the distortion parameter $\gamma$ is a constant $\O{1}$, then the number of messages $m$ sent by each player in \thmref{thm:main} is a constant. 
Moreover, under these settings, \thmref{thm:main} recovers the guarantees in the standard shuffle model from~\cite{BalleBGN20,GhaziMPV20}, though we remark that more communication efficient protocols~\cite{GhaziKMPS21} are known in the standard shuffle model across more general settings. 
Regardless, we again emphasize that the privacy and utility guarantees of the protocol are independent of the distortion parameter $\gamma$. 

\subsection{Overview of our Techniques}
In this section, we describe both our protocol for private real summation in the $\gamma$-I-shuffle model and the corresponding analysis for correctness and privacy. 

A natural starting point is the recent framework by~\cite{ZhouS22,ZhouSCM23}, which achieves amplification of privacy using \emph{differentially oblivious} (DO) shufflers that nearly match amplification of privacy results using fully anonymous shufflers~\cite{ErlingssonFMRTT19,BalleBGN19,CheuSUZZ19,FeldmanMT21}. 
Unfortunately, the framework crucially uses LDP protocols, which are known to not give optimal bounds even with fully anonymous shufflers. 
For instance, \cite{BalleBGN20,CheuSUZZ19,BalleBGN19b} showed that any single-message shuffled protocol for summation based on LDP protocols must exhibit mean squared error $\Omega(n^{1/3})$ or absolute error $\Omega(n^{1/6})$.

Another natural approach is to adapt recent works for private real summation in the shuffle model, e.g.,~\cite{GhaziMPV20,GhaziKMPS21}. 
One challenge in generalizing these proofs is that they often leverage the fully anonymous shuffler by analyzing a random sample from an alternate view of the output of the local randomizers, which often have some algebraic or combinatorial interpretation that facilitates the proof of specific desirable properties. 
However, these properties often seem substantially more difficult to prove once the symmetry of the fully anonymous shuffler is lost. 
In fact, we do not even know the mass that the $\gamma$-imperfect shuffler places on each permutation. 

\paragraph{From private real summation to statistical security of summation on fixed fields.}
We first use an observation from~\cite{BalleBGN20} that reduces the problem of private real summation to the problem of private summation on a fixed field of size $q$, so that each user has an input $x_i\in\mathbb{F}_q$ for all $i\in[n]$. 
We then consider the well-known split-and-mix protocol~\cite{IshaiKOS06}, where each user $i$ outputs a set of $m$ messages $x_{i,1},\ldots,x_{i,m}\in\mathbb{F}_q$ uniformly at random conditioned on $x_{i,1}+\ldots+x_{i,m}=x_i$. 
For the private summation on a fixed field problem, we adapt a well-known reduction~\cite{BalleBGN20} for the split-and-mix protocol  in the shuffle DP model to the notion of statistical security in the $\gamma$-I-shuffle model. 
Statistical security demands small total variation between the output of a protocol on input $x$ and input $x'$, if $\sum_{i=1}^n x_i=\sum_{i=1}^n x'_i$. 
In other words, it suffices to show that the output distribution looks ``similar'' for two inputs with the same sum. 
See \defref{def:sigma:security} for a formal definition of statistical security. 

To show statistical security, we first upper-bound the total variation distance in terms of the probability that two independent instances of the same protocol with the \emph{same} input give the same output. 
Balle et al.~\cite{BalleBGN20} use a similar approach, but then utilizes the symmetry of the fully anonymous shuffler to further upper-bound this quantity in terms of the probability that $\vec{\calR}(\vec{\sfX})=\calS\circ\vec{\calR}'(\vec{\sfX})$, where $\vec{\sfX}=(x_1,\ldots,x_n)$ is the input vector, $\vec{\calR}$ and $\vec{\calR}'$ are independent instances of the local randomizer, and $\calS$ is an instance of the uniform shuffler. 
We do not have access to such symmetries in the $\gamma$-I-shuffle model or even explicit probabilities that the $\gamma$-imperfect shuffler places on each permutation. 

\paragraph{Connected components on a communication graph.}
Instead, we first upper-bound the total variation distance by $\vec{\calR}(\vec{\sfX})=\calS^{-1}\circ\calS'\circ\vec{\calR}'(\vec{\sfX})$, where $\calS^{-1}$ is the inverse of an instance of a $\gamma$-imperfect shuffle and $\calS'$ is an independent instance of the same $\gamma$-imperfect shuffle. 
Intuitively, $\vec{\calR}(\vec{\sfX})$ and $\calS^{-1}\circ\calS'\circ\vec{\calR}'(\vec{\sfX})$ can look very different if there exists a large number of users whose messages are not shuffled with those of other users. 
Formally, this can be captured by looking at the number of connected components in the communication graph of $\calS^{-1}\circ\calS'\circ\vec{\calR}'(\vec{\sfX})$, so that there exists an edge connecting users $i$ and $j$ if the protocol swaps one of their messages. 
Hence, evaluating the number of connected components in the communication graph is closely related to analyzing the probability that there is no edge between $S$ and $[n]\smallsetminus S$, for a given set $S\subseteq[n]$. 

Although this quantity would be somewhat straightforward to evaluate for a uniform shuffler~\cite{BalleBGN20}, it seems more challenging to evaluate for $\gamma$-imperfect shufflers, since we do not have explicit probabilities for each permutation. 
Therefore, we develop a novel coupling argument to relate the probability that there is no edge between $S$ and $[n]\smallsetminus S$ in the $\gamma$-I-shuffle model to the probability of this event in the shuffle model. 
In particular, a specific technical challenge that our argument handles is when both $S$ and $[n]\smallsetminus S$ have large cardinality, because then there can be a permutation $\pi$ that swaps many coordinates while still leaving $S$ and $[n]\smallsetminus S$ disconnected. 
However, if we simply relate the probability of $\Pi$ in the shuffle and the $\gamma$-I-shuffle model, we incur a gap of $e^{t\cdot\gamma}$, where $\gamma$ is the distortion parameter and $t$ is the number of swaps by $\Pi$, which can have size $\Omega(n)$. 
Thus without additional care, this gap can overwhelm the probability achieved from the coupling argument. 
We circumvent this issue by considering a subset of $S$ with size $k$ and coupling the ``good'' permutations in the shuffle and the $\gamma$-I-shuffle model, which results in a smaller gap of $e^{k\cdot\gamma}$. 
For more details, see \lemref{lem:connect:one:prob}. 

\paragraph{Putting things together.}
At this point, we are almost done. 
Unfortunately, our coupling only addresses the case where a single imperfect shuffle is performed on a local randomizer, but we require the bound for the composition $\calS^{-1}\circ\calS'\circ\vec{\calR}'(\vec{\sfX})$, which seems significantly more challenging because communication between users $i$ and $j$ under $\calS'$ may be ``erased'' by $\calS^{-1}$. 
Instead, we show a simple observation for $\gamma$-imperfect shuffling, which states that if $\calS,\calS'$ are two shufflers such that $\calS$ is a $\gamma$-imperfect shuffler, then $\calS'\circ\calS$ is a $\gamma$-imperfect shuffler. 
This statement, presented in \lemref{lem:do:postprocess}, can be considered as a post-processing preservation property of $\gamma$-imperfect shuffling. 
In light of this statement, we can now view $\calS^{-1}\circ\calS'\circ\vec{\calR}'(\vec{\sfX})$ as a single $\gamma$-imperfect shuffler applied to the local randomizer, and use our new results upper-bounding the number of connected components in the resulting communication graph to ultimately show $\sigma$-security. 

Our analysis crucially utilizes the decomposition of the $\gamma$-imperfect shuffler on $m$ messages across $n$ users, i.e., the $m$-parallel shuffling model, by first considering the communication graph induced by a single round of shuffling between the $n$ users, and then scaling the effects $m$ times. 
Considering the model where all the $mn$ messages are shuffled together using a single shuffler will likely need a separate approach in the analysis. 

\subsection{Preliminaries}
For an integer $n>0$, we define $[n]:=\{1,\ldots,n\}$. 

\begin{definition}[Total variation distance]
Given probability measures $\mu,\nu$ on a domain $\Omega$, their total variation distance is defined by
\[\TVD(\mu,\nu)=\frac{1}{2}\|\mu-\nu\|_1=\frac{1}{2}\sum_{x\in\Omega}\lvert\mu(x)-\nu(x)\rvert.\]
\end{definition}

\begin{definition}[$\sigma$-security]
\deflab{def:sigma:security}
Given a security parameter $\sigma>0$, a protocol $\calP$ is $\sigma$-secure for computing a function $f:\calX^n\to Z$ if, for any $x,x'\in\calX^n$ such that $f(x)=f(x')$, we have
\[\TVD(\calP(x),\calP(x'))\le2^{-\sigma}.\]
\end{definition}

Recall the following two well-known properties of differential privacy:
\begin{theorem}[Basic Composition of differential privacy, e.g.,~\cite{DworkR14}]
\thmlab{thm:adv:comp}
Let $\eps, \delta \geq 0$. 
Any mechanism that permits $k$ adaptive interactions with mechanisms that preserve $(\eps,\delta)$-differential privacy guarantees $(k\eps,k\delta)$-differential privacy. 
\end{theorem}

\begin{theorem}[Post-processing of differential privacy~{\cite{DworkR14}}]
\thmlab{thm:dp:comp}
Let $\calM:\calU^*\to X$ be an $(\eps,\delta)$-differential private algorithm. 
Then, for any arbitrary random mapping $g:X\to X'$, we have that $g(\calM(x))$ is $(\eps,\delta)$-differentially private. 
\end{theorem}

We use $\DLap(\alpha)$ to denote the discrete Laplace distribution, so that $Z\sim\DLap(\alpha)$ has the probability mass function $\PPr{Z=k}\propto\alpha^{|k|}$ for $k \in \bbZ$. 
We use $\Polya(r, p)$ to denote the Polya distribution with parameter $r > 0, p \in (0, 1)$, which induces the probability density function $k \mapsto \binom{k + r - 1}{k} p^k(1 - p)^r$ for $k \in \bbZ_{\geq 0}$. 
We require the following equivalence between a discrete Laplacian random variable and the sum of a differences of Polya random variables.
\begin{fact}
\factlab{fact:polya:dlap}
Let $x_1,\ldots,x_n,y_1,\ldots,y_n\sim\Polya\left(\frac{1}{n},\alpha\right)$. 
Then $z=\sum_{i=1}^n(x_i-y_i)\sim\DLap(\alpha)$.  
\end{fact}

We also require the following property about randomized rounding.
\begin{lemma}
\lemlab{lem:round:err}
\cite{BalleBGN19}
Given a precision $p\ge 1$, let $x_1,\ldots,x_n\in\mathbb{R}$ and $y_i=\lfloor x_ip\rfloor+\Ber(x_ip-\lfloor x_ip\rfloor)$ for each $i\in[n]$. 
Then 
\[\Ex{\left(\sum_{i=1}^n\left(x_i-\frac{y_i}{p}\right)\right)^2}\le\frac{n}{4p^2}.\]
\end{lemma}

\subsection{Related Work}
To amplify the privacy in the shuffle model, the trusted shuffler is the key component of the shuffle model, which in some sense only shifts the point of vulnerability from the curator to the shuffler, particularly in the case where the shuffler may be colluding with the curator. 
Hence among the various relaxations for distributed DP protocols, e.g.~\cite{BalleKMTT20,CheuY23}, the DO shuffle model has been recently proposed~\cite{ShiW21,GordonKLX22} to permit some differentially private leakage in the shuffling stage, called a DO shuffle.  
In fact, \cite{ShiW21,GordonKLX22} showed that DO-shuffling can be more eﬀicient to achieve than a fully anonymous shuffle while \cite{ZhouS22,ZhouSCM23} showed that locally private protocols can be used in conjunction with a DO shuffler to achieve almost the same privacy amplification bounds as with a fully anonymous shuffler, up to a small additive loss resulting from the DO shuffle. 
However, the best known results in the shuffle model of DP do not utilize LDP protocols, and thus cannot directly be applied in the framework of \cite{ZhouS22,ZhouSCM23}.

\section{A Simple Reduction}
In this section, we briefly describe a simple reduction for showing amplification of privacy for imperfect shuffling. 
The result can be viewed as in the same spirit as similar privacy amplification statements, e.g., \cite{FeldmanMT21,ZhouS22,FeldmanMT23}, but for imperfect shuffling. 
In particular, the following well-known result achieves privacy amplification for local randomizers in the shuffle model:

\begin{theorem}
\cite{FeldmanMT21}
\thmlab{thm:shuffle:priv:amp}
For any domain $\calD$ and $i\in[n]$, let $\calR^{(i)}:\calX^{(1)}\times\ldots\times\calX^{(i-1)}\times\calD\to\calX^{(i)}$, where $\calX^{(i)}$ is the range space of $\calR^{(i)}$, such that $\calR^{(i)}(z_{1:i-1},\cdot)$ is an $\eps_0$-DP local randomizer for all values of auxiliary inputs $z_{1:i-1}\in\calX^{(1)}\times\ldots\times\calX^{(i-1)}$. 
Let $\calA_s:\calD^n\to\calX^{(1)}\times\ldots\times\calX^{(n)}$ be the algorithm that given a dataset $x_{1:n}\in\calD^n$, samples a uniform random permutation $\pi$ over $[n]$ and sequentially computes $z_i=\calR^{(i)}(z_{1:i-1},x_{\pi(i)})$ for $i\in[n]$ and outputs $z_{1:n}$. 
Then for any $\delta\in[0,1]$ such that $\eps_0\le\log\left(\frac{n}{16\log(2/\delta)}\right)$, $\calA_s$ is $(\eps,\delta)$-DP for
\[\eps\le\log\left(1+\frac{e^{\eps_0}-1}{e^{\eps_0}+1}\left(\frac{8\sqrt{e^{\eps_0}\log(4/\delta)}}{\sqrt{n}}\right)\right).\]
\end{theorem}

We would like to show privacy amplification statements for the imperfect shuffle model that are qualitatively similar to \thmref{thm:shuffle:priv:amp}. 
To that end, we first recall the following definition of differentially oblivious shufflers. 

\begin{definition}[Differentially Oblivious Shuffle, e.g., \cite{ChanCMS22,ZhouSCM23}]
A shuffle protocol is $(\eps,\delta)$-differentially oblivious if for all adversaries $\calV$, all $\pi,\pi'\in\Pi$, and all subsets $S$ of the view space,
\[\PPr{\View^{\calV}(\pi)\in S}\le e^{\eps\cdot\Swap(\pi,\pi')}\PPr{\View^{\calV}(\pi')\in S}+\delta.\]
\end{definition}

\cite{ZhouS22} showed that differentially oblivious shufflers also amplify privacy. 

\begin{theorem}[Theorem 1 in \cite{ZhouS22}]
\thmlab{thm:do:shuffle:amplify}
For any domain $\calD$ and range space $\calX$, $i\in[n]$, let $\calR^{(1)},\ldots,\calR^{(n)}:\calD\to\calX$ be $\eps_0$-DP local randomizers and let $\calA_s$ be a $(\eps_1,\delta_1)$-DO shuffler. 
Then the composed protocol $\calA_s(\calR^{(1)},\ldots,\calR^{(n)})$ is $(\eps+\eps_1,\delta+\delta_1)$-DP for
\[\eps=\O{\frac{(1-e^{\eps_0})e^{\eps_0/2}\sqrt{\log(1/\delta)}}{\sqrt{n}}}.\]
\end{theorem}

It turns out that imperfect shufflers can be parametrized by differentially oblivious shufflers. 
That is, imperfect shufflers are a specific form of differentially oblivious shufflers. 
Therefore, we can immediately apply the previous statement to obtain the following statement for privacy amplification for imperfect shufflers. 

\begin{theorem}
For any domain $\calD$ and range space $\calX$, $i\in[n]$, let $\calR^{(1)},\ldots,\calR^{(n)}:\calD\to\calX$ be $\eps_0$-DP local randomizers and let $\calA_s$ be a $\gamma$-imperfect shuffler. 
Then the composed protocol $\calA_s(\calR^{(1)},\ldots,\calR^{(n)})$ is $(\eps+\gamma,\delta)$-DP for
\[\eps=\O{\frac{(1-e^{\eps_0})e^{\eps_0/2}\sqrt{\log(1/\delta)}}{\sqrt{n}}}.\]
\end{theorem}
\begin{proof}
By the definition of $\gamma$-imperfect shuffle, we have that for all $\pi,\pi'\in\Pi$,
\[\PPr{\calS=\pi}\le e^{\gamma \cdot \Swap(\pi,\pi')}\PPr{\calS=\pi'}.\]
Since no additional information is leaked by the shuffler, then for all adversaries $\calV$ and all subsets $S$ of the view space,
\[\PPr{\View^{\calV}(\pi)\in S}\le e^{\gamma\cdot\Swap(\pi,\pi')}\PPr{\View^{\calV}(\pi')\in S}.\]
In other words, the $\gamma$-imperfect shuffler is a $(\gamma,0)$-DO shuffler. 
Thus by \thmref{thm:do:shuffle:amplify}, the composed protocol $\calA_s(\calR^{(1)},\ldots,\calR^{(n)})$ is $(\eps+\gamma,\delta)$-DP for
\[\eps=\O{\frac{(1-e^{\eps_0})e^{\eps_0/2}\sqrt{\log(1/\delta)}}{\sqrt{n}}}.\]
\end{proof}

\section{Differentially Private Summation}
In this section, we first introduce the structural statements necessary to argue privacy for the standard split-and-mix protocol~\cite{IshaiKOS06}. 
We then assume correctness of these statements, deferring their proofs to subsequent sections, and we prove the guarantees of \thmref{thm:main}. 
We also give an application to private vector aggregation as a simple corollary of \thmref{thm:main}. 

We first relate differentially private protocols for summation under a $\gamma$-imperfect shuffler to $\sigma$-secure protocols. 
Lemma 4.1 in \cite{BalleBGN19} showed this relationship for uniform shufflers. 
It turns out their proof extends to $\gamma$-imperfect shufflers as well. 
For the sake of completeness, we include the proof in \appref{app:extra:proofs}. 
\begin{restatable}{lemma}{lemsecuretodp}[Lemma 4.1 in~\cite{BalleBGN19}]
\lemlab{lem:secure:to:dp}
Given a $\sigma$-secure protocol $\Xi$ in the $\gamma$-I-shuffle model for $n$-party private summation on $\mathbb{Z}_q$ such that each player sends $f(n,q,\sigma)$ bits of messages, there exists a $(\eps,(1+e^\eps)2^{-\sigma-1})$-differentially private protocol in the $\gamma$-I-shuffle model for $n$-party private summation on real numbers with expected absolute error $\O{\frac{1}{\eps}}$ such that each player sends $f(n,O(n^{3/2}),\sigma)$ bits of messages. 
\end{restatable}

In \secref{sec:split:mix}, we prove the following guarantees about the split-and-mix protocol from~\cite{IshaiKOS06}. 
\begin{restatable}{theorem}{thmsecurecc}
\thmlab{thm:secure:cc}
Let $n\ge 19$ and $\gamma\le\frac{\log\log n}{80}$ be a distortion parameter. 
For worst-case statistical security with parameter $\sigma$, it suffices to use
\[m=\O{e^{4\gamma}+\frac{e^{4\gamma}(\sigma+\log n)}{\log n}}\]
messages. 
Each message uses $\O{\log q}$ bits, for $q=\ceil{2n^{3/2}}$. 
\end{restatable}
\pasin{I don't understand this $\min\left(\frac{\log\log n}{80},\frac{\log n}{4m}\right)$ term, isn't the first term always smaller than the second term?}
\samson{Oops, yes}
By \lemref{lem:secure:to:dp} and \thmref{thm:secure:cc}, we have our main statement:
\thmmain*

\paragraph{Applications to private vector summation.}
An immediate application of our results is to the problem of private vector aggregation, where $n$ parties have vectors $\vec{x_1},\ldots,\vec{x_n}\in [0,1]^d$ and the goal is to privately compute $\vec{X}=\sum_{i=1}^n\vec{x_i}\in\mathbb{R}^d$. 
Given a protocol $\calP$ for private summation where $n$ players each send $m$ messages, the $n$ players can perform a protocol $\calP'$ for vector aggregation by performing $\calP$ on each of their $d$ coordinates. 
In particular, the $n$ players can first perform $\calP$ on the first coordinate of their vectors, then perform $\calP$ on the second coordinate of their vectors, and so on and so forth, by sending $md$ messages in total. 
Equivalently, the $n$ players can perform $\calP$ on a field of size $q^d$ rather than size $q$ and just send $m$ messages in total. 
However, the total communication size is still the same, because each message increases by a factor of $d$ due to the larger field size. 
Thus we consider the approach where the $n$ players perform $\calP$ on each of the $d$ coordinates. 

To argue privacy, we observe that the $n$ players run $d$ iterations of the protocol $\calP$, once for each of the coordinates. 
By composition of DP, i.e., \thmref{thm:adv:comp}, to guarantee $\eps$-privacy for the overall protocol, it suffices to run each of the $d$ iterations with privacy $\eps'=\frac{\eps}{d}$ and failure probability $\delta'=\frac{\delta}{d}$. 
By post-processing of DP, i.e., \thmref{thm:dp:comp}, the resulting vector where each coordinate is computed using the corresponding protocol is $(\eps,\delta)$-DP. 

Then as a corollary to \thmref{thm:main} with privacy parameter $\eps'=\frac{\eps}{d}$ and failure probability $\delta'=\frac{\delta}{d}$:
\begin{restatable}{theorem}{thmvector}
\thmlab{thm:vector}
Let $n\ge 19$, $d\ge 1$, $\eps>0$ be a (constant) privacy parameter, and $\gamma\le\frac{\log\log n}{80}$ be a distortion parameter. 
Then there exists an $(\eps,\delta)$-DP protocol for vector summation in the $\gamma$-I-shuffle model with expected absolute error $\O{\frac{d}{\eps}}$ per coordinate and 
\[m=\O{d\left(e^{4\gamma}+\frac{e^{4\gamma}(\log\frac{d}{\delta}+\log n)}{\log n}\right)}\]
messages per party. 
Each message uses $\O{\log q}$ bits, for $q=\ceil{2n^{3/2}}$. 
\end{restatable}

We remark that for certain regimes of $\eps$ and $\delta$, \thmref{thm:vector} can be easily improved by using \thmref{thm:main} and advanced composition, rather than basic composition. 

\section{Security of Split-and-Mix Protocol}
\seclab{sec:split:mix}
In this section, we prove the $\sigma$-security of the split-and-mix protocol. 
The proof largely attempts to follow the outline of the split-and-mix protocol analysis for private aggregation by~\cite{BalleBGN19}, which first reduces from worst-case input to average-case input and then analyzes the connectivity of the resulting communication graph induced by a uniform shuffle. 

We similarly first reduce from worst-case input to average-case input and then analyze the connectivity of the resulting communication graph induced by a $\gamma$-imperfect shuffle. 
The former appears in \secref{sec:worst:average} and the latter appears in \secref{sec:reduce:cc}. 

However, the main challenge is that the symmetric properties of the uniform shuffler is often crucially utilized in various steps of the approach. 
Unfortunately, these properties do not often seem to translate to $\gamma$-imperfect shufflers, where we might not even know the mass that is placed on each permutation. 
Thus we need to handle a number of technical challenges to recover qualitatively similar structural properties to the uniform shuffling model. 
Along the way, we show that the composition of two shufflers, where the inner shuffler is a $\gamma$-imperfect shuffler, is also a $\gamma$-imperfect shuffler with the same parameter, which can be interpreted as a post-processing statement for $\gamma$-imperfect shuffling. 

We first formally define the split-and-mix protocol:
\begin{definition}[Split-and-Mix Protocol, e.g., \cite{IshaiKOS06}]
Given an integer parameter $m\ge 1$, the $m$-message $n$-player split-and-mix protocol $\calP_{m,n}$ is defined as follows. 
Each player $i$ outputs a set of $m$ messages $x_{i,1},\ldots,x_{i,m}$ uniformly at random conditioned on $x_{i,1}+\ldots+x_{i,m}=x_i$. 
For each $j\in[m]$, the set of messages $x_{1,j},\ldots,x_{n,j}$ are then swapped according to a $\gamma$-imperfect shuffler $\calS^{(j)}$.
\end{definition}

\subsection{Worst-case to Average-case Reduction}
\seclab{sec:worst:average}
In this section, we show a reduction from worst-case input to average-case input. 
In other words, rather than analyze the split-and-mix protocol over the worst-case input, we show it suffices to analyze the expected performance of the split-and-mix protocol for a uniformly random input. 
The approach is nearly identical to that of \cite{BalleBGN20}, but they can further simplify their final expression due to the symmetric properties of the uniform shuffler, which do not hold for the $\gamma$-imperfect shuffler. 

Let $\calP_{m,n}$ denote the $m$-message $n$-player split-and-mix protocol and let $\tilde{\calP}_{m,n}$ be defined as follows. 
Each player $i$ outputs a set of $m+1$ messages $x_{i,1},\ldots,x_{i,m+1}$ uniformly at random conditioned on $x_{i,1}+\ldots+x_{i,m+1}=x_i$. 
For each $i\in[n]$, we use the notation $\calR_m(x_i)=(x_{i,1},\ldots,x_{i,m})$ to denote the choice of the $m$ messages for player $i$. 
Let $\mathbb{G}=\mathbb{F}_q$ and for $j\in[m]$, let $\calS^{(j)}:\mathbb{G}^n\to\mathbb{G}^n$ be independent shufflers. 
Then the output of $\tilde{\calP}_{m,n}$ is $\calS^{(j)}$ applied to the first $m$ messages of each player, concatenated with the unshuffled final message of each player, i.e.,
\[\tilde{\calP}_{m,n}(x_1,\ldots,x_n)=\calS^{(1)}(x_{1,1},\ldots,x_{n,1})\circ\ldots\circ\calS^{(m)}(x_{1,m},\ldots,x_{n,m})\circ x_{1,m+1},\ldots,x_{n,m+1}.\]

We first reduce the problem to average-case statistical security using the approach of Lemma 6.1 in \cite{BalleBGN20}. 
Formally, we say that a protocol $\calP_{m,n}$ provides average-case statistical security with parameter $\sigma$ if
\[\mathbb{E}_{\vec{\sfX},\vec{\sfX}'}[\TVD_{\mid\vec{\sfX},\vec{\sfX}'}(\calP_{m,n}(\vec{\sfX}),\calP_{m,n}(\vec{\sfX}'))]\le 2^{-\sigma},\]
where $\vec{\sfX}$ and $\vec{\sfX}'$ are each drawn uniformly at random from all pairs of vectors in $\mathbb{G}^n$ with the same sum. 
Here we use the notation $\TVD_{\mid\vec{\sfX},\vec{\sfX}'}$ to denote the total variation distance between two distributions conditioned on fixings of $\vec{\sfX}$ and $\vec{\sfX}'$.
\begin{lemma}
\lemlab{lem:avg:worst}
Suppose $\calP_{m,n}$ provides average-case statistical security with parameter $\sigma$, then $\calP_{m+1,n}$ and $\tilde{\calP}_{m,n}$ provide worst-case statistical security with parameter $\sigma$.
\end{lemma}
\begin{proof}
Let $\vec{x}$ and $\vec{x'}$ be a pair of vectors in $\mathbb{G}^n$ with the same sum. 
Given an output of $\tilde{\calP}_{m,n}(\vec{x})$, the protocol $\calP_{m+1,n}(\vec{x})$ can be simulated by using an additional application of $\calR_{m+1}$ to randomly permute the last message of each of the players according to the distribution of the $\gamma$-imperfect shuffle. 
Hence, 
\[\TVD(\calP_{m+1,n}(\vec{x}),\calP_{m+1,n}(\vec{x'}))\le\TVD(\tilde{\calP}_{m,n}(\vec{x}),\tilde{\calP}_{m,n}(\vec{x'})).\]
It thus suffices to upper bound the worst-case statistical security of $\tilde{\calP}_{m,n}$ by $\sigma$. 

The worst-case security of $\tilde{\calP}_{m,n}$ is reduced to the average-case security of $\tilde{\calP}_{m,n}$ by noting that the addition of the $(m+1)$-th message to each player can effectively be viewed as adding a random value to each player's input and thus transforming each input value $x_i$ into a uniformly random value in $\mathbb{G}$. 
More formally, consider the definition
\[\calR_{m+1}(x)=(\calR_m(x-\sfU),\sfU),\]
for $x\in\mathbb{G}$, where $\sfU$ is a uniformly random element of $\mathbb{G}$. 

Since $\vec{x}-\vec{\sfU}$ is a uniformly random vector in $\mathbb{G}^n$, then we can couple the randomness observed from two instances $\vec{\sfU},\vec{\sfU'}$ resulting from two independent executions of $\calP_{m, n}$ with two inputs having the same sum. 
Therefore,
\begin{align*}
\TVD(\tilde{\calP}_{m+1,n}(\vec{x}),\tilde{\calP}_{m+1,n}(\vec{x'}))&=\TVD((\calP_{m,n}(\vec{x}-\vec{\sfU}),\vec{\sfU}),(\calP_{m,n}(\vec{x'}-\vec{\sfU'}),\vec{\sfU'}))\\
&=\mathbb{E}_{\vec{\sfU},\vec{\sfU'}}[\TVD(\calP_{m,n}(\vec{x}-\vec{\sfU}),\calP_{m,n}(\vec{x'}-\vec{\sfU'}))]]\\
&=\mathbb{E}_{\vec{\sfX},\vec{\sfX}'}[\TVD(\calP_{m,n}(\vec{\sfX}),\calP_{m,n}(\vec{\sfX}'))],
\end{align*}
where $\vec{\sfX},\vec{\sfX}'$ are chosen uniformly at random conditioned on $\vec{\sfX}=\vec{\sfX}'+\vec{x}-\vec{x'}$. 
\end{proof}

We now upper bound the expected total variation distance between the two independent executions of the $\gamma$-imperfect shuffle, using an approach similar to Lemma C.1 in \cite{BalleBGN20}. 
\begin{lemma}
\lemlab{lem:avg:ub}
Let $\vec{\sfX}$ and $\vec{\sfX'}$ be drawn uniformly at random from all pairs of vectors in $\mathbb{G}^n$ with the same sum, noting that $\vec{\sfX}$ and $\vec{\sfX'}$ are not independent.  
For two independent executions $\calP_{m,n}$ and $\calP'_{m,n}$ of the $\gamma$-imperfect shuffle, 
\[\mathbb{E}_{\vec{\sfX},\vec{\sfX}'}[\TVD_{\mid\vec{\sfX},\vec{\sfX}'}(\calP_{m,n}(\vec{\sfX}),\calP_{m,n}(\vec{\sfX}'))]\le\sqrt{q^{mn-1}\PPr{\calP_{m,n}(\vec{\sfX})=\calP'_{m,n}(\vec{\sfX})}-1}.\]
\end{lemma}
\begin{proof}
We write $\calP$ and $\calP'$ as shorthand for $\calP_{m,n}$ and $\calP'_{m,n}$, respectively. 
Let $\vec{\sfV}$ be a uniformly random vector drawn from $\mathbb{G}^{mn}$, conditioned on $\vec{\sfV}$ having the same sum as $\vec{\sfX}$ and $\vec{\sfX'}$. 
Then by the triangle inequality, 
\begin{align*}
\mathbb{E}_{\vec{\sfX},\vec{\sfX}'}[\TVD_{\mid\vec{\sfX},\vec{\sfX}'}(\calP(\vec{\sfX}),\calP(\vec{\sfX}'))]&\le\mathbb{E}_{\vec{\sfX},\vec{\sfX}'}[\TVD_{\mid\vec{\sfX},\vec{\sfX}'}(\calP(\vec{\sfX}),\vec{\sfV})+\TVD_{\mid\vec{\sfX},\vec{\sfX}'}(\vec{\sfV},\calP(\vec{\sfX}'))]\\
&=\mathbb{E}_{\vec{\sfX}}[\TVD_{\mid\vec{\sfX}}(\calP(\vec{\sfX}),\vec{\sfV})]+\mathbb{E}_{\vec{\sfX}'}[\TVD_{\mid\vec{\sfX}'}(\vec{\sfV},\calP(\vec{\sfX}'))]\\
&=2\mathbb{E}_{\vec{\sfX}}[\TVD_{\mid\vec{\sfX}}(\calP(\vec{\sfX}),\vec{\sfV})].
\end{align*}
Moreover, considering the distribution over $\vec{\sfV}$,
\begin{align*}
2\TVD_{\mid\vec{\sfX}}(\calP(\vec{\sfX}),\vec{\sfV})&=\sum_{\vec{v}\in\mathbb{G}^{mn}}\left\lvert\PPr{\calP(\vec{\sfX})=\vec{v}}-\PPr{\vec{\sfV}=\vec{v}}\right\rvert\\
&=\sum_{\vec{v}\in\mathbb{G}^{mn}, \sum\vec{v}=\sum\vec{\sfX}}|\PPr{\calP(\vec{\sfX})=\vec{v}}-q^{1-mn}|\\
&=q^{mn-1}\EEx{\vec{\sfV}}{\left\lvert\PPr{\calP(\vec{\sfX})=\vec{\sfV}}-q^{1-mn}\right\rvert}.
\end{align*}
Since $\vec{\sfV}$ is a uniformly random vector from $\mathbb{G}^{mn}$ with its sum being equal to that of $\vec{\sfX}$, then for the random variable $\calZ:=\calZ(\sfX,\sfV):=\PPr{\calP(\vec{\sfX})=\vec{\sfV}}$, we have $\Ex{\calZ}=q^{1-mn}$. 
Therefore, 
\[2\TVD_{\mid\vec{\sfX}}(\calP(\vec{\sfX}),\vec{\sfV})\le q^{mn-1}\mathbb{E}[|\calZ-\mathbb{E}[\calZ]|].\]
By convexity,
\[\mathbb{E}[|\calZ-\mathbb{E}[\calZ]|]\le\sqrt{\mathbb{E}[\calZ^2]}.\]
Since we have
\begin{align*}
\mathbb{E}_{\vec{\calV}}[\calZ^2]&=q^{1-mn}\sum_{\vec{v}\in\mathbb{G}^{mn}, \sum\vec{v}=\sum\vec{\sfX}}\PPr{\calP(\vec{\sfX})=\vec{v}}^2\\
&=q^{1-mn}\PPr{\calP(\vec{\sfX})=\calP'(\vec{\sfX})},
\end{align*}
we thus have
\begin{align*}
\mathbb{E}_{\vec{\sfX},\vec{\sfX}'}[\TVD_{\mid\vec{\sfX},\vec{\sfX}'}(\calP(\vec{\sfX}),\calP(\vec{\sfX}')]&\le2\TVD_{\mid\vec{\sfX}}(\calP(\vec{\sfX}),\vec{\sfV})\\
&\le q^{mn-1}\mathbb{E}_{\calV(\vec{\sfX}')}[|\PPr{\calP(\vec{\sfX})=\vec{\sfV}}-q^{1-mn}|]\\
&\le\sqrt{q^{mn-1}\PPr{\calP_{m,n}(\vec{\sfX})=\calP'_{m,n}(\vec{\sfX})}-1}. & &\qedhere
\end{align*}
\end{proof}
We note that the probability that two independent executions of the protocol can be decomposed into the split protocol and the mix protocol as follows. 
By comparison, Lemma C.2 in \cite{BalleBGN20} was able to prove a simpler relationship by leveraging properties of their symmetric shuffler, which we do not have for an imperfect shuffler. 
\begin{lemma}
\lemlab{lem:avg:two:shuf}
Let $\calR_{m,n}$ and $\calR'_{m,n}$ denote two independent executions of the split protocol in $\calP_{m,n}$ so that $\calP_{m,n}=\calS_{m,n}\circ\calR_{m,n}$. 
Then
\[\PPr{\calP_{m,n}(\vec{\sfX})=\calP'_{m,n}(\vec{\sfX})}=\PPr{\calR_{m,n}(\vec{\sfX})=\calS^{-1}_{m,n}\circ\calS'_{m,n}\circ\calR'_{m,n}(\vec{\sfX})}.\]
\end{lemma}
\begin{proof}
Note that 
\begin{align*}
\PPr{\calP_{m,n}(\vec{\sfX})=\calP'_{m,n}(\vec{\sfX})}&=\PPr{\calS_{m,n}\circ\calR_{m,n}(\vec{\sfX})=\calS'_{m,n}\circ\calR'_{m,n}(\vec{\sfX})}\\
&=\PPr{\calR_{m,n}(\vec{\sfX})=\calS^{-1}_{m,n}\circ\calS'_{m,n}\circ\calR'_{m,n}(\vec{\sfX})}.
\end{align*}
\end{proof}

From \lemref{lem:avg:ub} and \lemref{lem:avg:two:shuf}, we have
\begin{lemma}
\lemlab{lem:exp:tvd}
For two independent executions $\calP_{m,n}$ and $\calP'_{m,n}$ of the split-and-mix protocol with a $\gamma$-imperfect shuffler,
\[\mathbb{E}_{\vec{\sfX},\vec{\sfX}'}[\TVD(\calP_{m,n}(\vec{\sfX}),\calP_{m,n}(\vec{\sfX}'))]\le\sqrt{q^{mn-1}\PPr{\calR_{m,n}(\vec{\sfX})=\calS^{-1}_{m,n}\circ\calS'_{m,n}\circ\calR'_{m,n}(\vec{\sfX})}-1}.\]
\end{lemma}

\subsection{Reduction to Connected Components}
\seclab{sec:reduce:cc}
In this section, we prove the following general statement upper bounding the probability that the shuffler $\calS^{-1}_{m,n}\circ\calS'_{m,n}(\cdot)$ on the output of a randomizer achieves the same output as an independent instance of the randomizer by the expectation of a quantity relating to the number of connected components in the communication graph of the shuffler $\calS^{-1}_{m,n}\circ\calS'_{m,n}(\cdot)$. 
Specifically, we can view a protocol $\calP_{m,n}$ that is an ordered pair $\pi_1,\ldots,\pi_m$, where $\pi_j$ is a permutation on $[n]$ for each $j\in[m]$, so that in each round $j\in[m]$, user $i\in[n]$ sends a message to user $\pi_j(i)$. 

Then we can define the \emph{communication graph} for the multi-message shuffle protocol $\calP_{m,n}$ as follows. 
The graph $G$ consists of $n$ vertices, which we associate with $[n]$, corresponding to the players $[n]$ participating in the protocol $\calP_{m,n}$. 
We add an edge between vertices $i$ and $j$ if player $i$ passes one of their $m$ messages to player $j$. 

The following proof is the same as Lemma C.4 in \cite{BalleBGN20}. 
\begin{lemma}
\lemlab{lem:prob:shuf:exp:graph}
Let $G$ be the graph on $n$ vertices formed the communication graph of the shuffle $\calS^{-1}\circ\calS'$. 
Let $C(G)$ be the number of connected components of $G$. 
Then
\[\PPr{\vec{\calR}(\vec{\sfX})=\calS^{-1}\circ\calS'\circ\vec{\calR}'(\vec{\sfX})}\le\Ex{q^{C(G)-mn}}.\]
\end{lemma}
\begin{proof}
By the law of total expectation,
\[\PPr{\vec{\calR}(\vec{\sfX})=\calS^{-1}\circ\calS'\circ\vec{\calR}'(\vec{\sfX})}=\Ex{\PPr{\vec{\calR}(\vec{\sfX})=\calS^{-1}\circ\calS'\circ\vec{\calR}'(\vec{\sfX})\mid\calS,\calS'}}.\]
Thus for the graph $G$ conditioned on $\calS$ and $\calS'$, it suffices to show that 
\[\PPr{\vec{\calR}(\vec{\sfX})=\calS^{-1}\circ\calS'\circ\vec{\calR}'(\vec{\sfX})\mid\calS,\calS'}=q^{C(G)-mn}.\]
Note that $C(G)$ depends on the choices of $\calS$ and $\calS'$ but we omit these dependencies in the notation for the sake of presentation. 
Recall that $\calP_{m,n}(\vec{\sfX})=\calS_{m,n}\circ\vec{\calR}_{m,n}(\vec{\sfX})$ is currently indexed so that the first message of each player after the shuffle protocol completes are the first $n$ indices, followed by the second message of each of the $n$ players and so forth. 
We thus define a re-indexing permutation $\psi:[mn]\to[mn]$ to that the $m$ messages of the first player will be the first $m$ indices, followed by the $m$ messages of the second player and so forth. 
That is,
\[\psi(j)=\left\lfloor\frac{j-1}{m}\right\rfloor+n(j-1\bmod{n})+1.\]
Let $\sfW,\sfW'\in\mathbb{G}^{mn}$ be defined so that $\sfW_j=\psi(\vec{\calR}(\vec{\sfX}))_j$ and $\sfW'_j=\psi(\calS^{-1}\circ\calS'\circ\vec{\calR}'(\vec{\sfX}))_j$. 
The task then becomes to show that
\[\PPr{\sfW=\sfW'\mid\calS,\calS'}=q^{C(G)-mn}.\]
Toward that end, for each $j\in[mn]$, we define $\calE_j$ to be the event that $\sfW_j=\sfW'_j$ and $p_j=\PPr{\calE_j\mid\calE_1,\ldots,\calE_{j-1}}$, so that
\[\PPr{\sfW=\sfW'\mid\calS,\calS'}=\prod_{j=1}^mn p_j.\]

Firstly, consider the messages that are not the last message by a particular player, i.e., consider the values of $j\in[mn]$ that are not divisible by $m$. 
Observe that conditioning on fixed values of $\vec{\sfX}$ and $\vec{\calR}'$, as well as the events $\calE_1,\ldots,\calE_{j-1}$, the value of $\sfW_j$ remains uniformly distributed and has probability $q^{-1}$ of being equal to to $\sfW'_{j}$. 
Hence, we have $p_j=q^{-1}$. 

For the cases where $j$ is divisible by $m$, we further consider two subcases. 
In particular, we consider the case where $j$ is the largest index in $C_j$ and the case where $j$ is not the largest index in $C_j$, where $C_j$ is the set of vertices in the same connected component as $j$ in $G$. 

In the first subcase, the multisets of $\sfW'$ and $\vec{\calR}'(\sfX')$ restricted to $C_i$ are the same and thus the multisets of the summands are the same, so that
\[\sum_{i\mid C_i=C_j}\sfW'_i=\sum_{i\mid C_i=C_j}\psi(\vec{\calR}'(\sfX'))_i.\]
Moreover, since the indices corresponding to all messages of a fixed player are in the same connected component, then
\[\sum_{i\mid C_i=C_j}\psi(\vec{\calR}'(\sfX'))_i=\sum_{i\mid C_i=C_j}\sfW_i.\]
Finally, we have that conditioning on $\calE_1,\ldots,\calE_{j-1}$ and the fact that $j$ is the largest index in $C_j$,
\[\sum_{i\mid C_i=C_j,i\neq j}\sfW'_i=\sum_{i\mid C_i=C_j,i\neq j}\sfW_i.\]
Therefore, we have $p_j=1$.

For the second subcase, we shall show that $p_j=q^{-1}$. 
Let $\calT$ be the subset of $(\sfW,\sfW')\in\mathbb{G}^{2mn}$ that are consistent with $\calE_1,\ldots,\calE_{j-1}$ and a fixed value of $\vec{\sfX}$. 
We show there exists a homomorphism $\phi:\mathbb{G}\to\mathbb{G}^{2mn}$ that maps from $g\in\mathbb{G}$ to a $u_g\in\mathbb{G}^{2mn}$ with a specific property to be defined.  
We then consider the action of $\mathbb{G}^{2mn}$ on itself by addition of $u_g$. 
Then the property of $\phi$ that we show is that $u_g$ fixes $\calT$ and $\sfW_j$ but adds $g$ to $\sfW'_j$. 
Consider the partitioning of $\calT$ into equivalence classes where two elements of $\calT$ are equivalent if they are equal under addition by $u_g$ for some $g$. 
Then the homomorphism induces a partitioning of $\calT$ into subsets of size $q$ such that each subset contains exactly one element for which $\calE_j$ holds. 
Since each value of $\calT$ is equally probable, it then follows that $p_j=q^{-1}$ as desired. 

We now define the homomorphism $\phi$ as follows. 
Since there exists a path in $G$ from the vertex with the $j$-th message to a higher index vertex, then there exists some path parameter $\ell$ and a corresponding path $(a_1,b_1,\ldots,a_\ell,b_\ell,a_{\ell+1})$ such that the following hold. 
Firstly, each of the terms $a_i,b_i$ are elements of $[mn]$ that will ultimately map to indices of elements in $\mathbb{G}^{mn}$. 
Secondly, for all $i\in[\ell]$, we have $\pi(b_i)=a_i$ for the permutation $\pi$ induced by the $m$ message $n$ player protocol and moreover, $b_i$ and $a_{i+1}$ correspond to the same vertex. 
Finally, it holds that $a_1=j$, $b_\ell>j$, $a_i\neq a_{i'}$ for any $i\neq i'$, and $b_i<j$ for all $i<\ell$. 
Then we implicitly define the homomorphism $\phi$ by defining $u_g$ to be the element of $\mathbb{G}^{2mn}$ with the value $g$ in the entries $a_2,\ldots,a_{\ell+1},b_1+mn,\ldots,b_\ell+mn$ and the identity $0$ in all other coordinates, where we recall that the elements $a_i$ and $b_i$ correspond to indices of elements in $\mathbb{G}^{mn}$. 

We observe that the group action of addition by $u_g$ does not affect the realization of $\calE_1,\ldots,\calE_{j-1}$ since $\sfW_{a_i}$ and $\sfW'_{a_i}=\vec{\calR}'(\vec{\sfX})_{b_i}$ are increased by exactly the same amount by $u_g$, except for the case when $i=1$ or $i=\ell+1$. 
However, note that $a_i\ge j$ for both of the cases where $i=1$ and $i=\ell+1$, which does not affect the realization of $\calE_1,\ldots,\calE_{j-1}$. 
Hence, $u_g$ has the desired properties and so it follows that $p_j=q^{-1}$. 

Therefore, conditioned on any fixed realization of $\calS$, we have that 
\[\prod_{j=1}^{mn} p_j=q^{C(G)-mn},\]
so that in summary
\[\PPr{\vec{\calR}=\calS^{-1}\circ\calS'\circ\vec{\calR}'(\vec{\sfX})}\le\mathbb{E}[q^{C(G)-mn}].\]
\end{proof}
We remark that the statement of \lemref{lem:prob:shuf:exp:graph} holds even for a general shuffler $\calS$ with the corresponding communication graph, rather than the specific shuffler $\calS^{-1}_{m,n}\circ\calS'_{m,n}(\cdot)$. 

We now show that the composition of two shufflers, where the inner shuffler is a $\gamma$-imperfect shuffler, is also a $\gamma$-imperfect shuffler with the same parameter. 
\begin{lemma}
\lemlab{lem:do:postprocess}
Let $\calS,\calS'$ be two shufflers such that $\calS$ is a $\gamma$-imperfect shuffler. 
Then, $\calS'\circ\calS$ is a $\gamma$-imperfect shuffler.
\end{lemma}
\begin{proof}
Let $\calS'$ be an arbitrary shuffler and $\calS$ be a $\gamma$-imperfect shuffler. Then, for any $\pi, \pi' \in \Pi$,
\begin{align*}
\PPr{\calS' \circ \calS = \pi} &= \PPr{\calS = (\calS')^{-1} \circ \pi}\\
&\leq e^{\gamma \cdot \Swap((\calS')^{-1} \circ \pi, (\calS')^{-1} \circ \pi')}\PPr{\calS = (\calS')^{-1} \circ \pi'} \\
&= e^{\gamma \cdot \Swap(\pi, \pi')}\PPr{\calS' \circ \calS = \pi'}.
\end{align*}
Thus, $\calS'\circ\calS$ is a $\gamma$-imperfect shuffler. 
\end{proof}


We now show a few structural statements that upper bound the probability that there exists no edge from a set $S\subset[n]$ to $[n]\setminus S$ for a communication graph induced by a $\gamma$-imperfect shuffler. 
\begin{lemma}
\lemlab{lem:connect:one:prob:first}
Let $G$ be the communication graph of a $\gamma$-imperfect shuffler (on an $n$-player $m$-message protocol). 
For a fixed set $S$ with size $s$, the probability that there exists no edge from $S$ to $[n]\setminus S$ in $G$ is at most $e^{2sm\gamma}\binom{n}{s}^{-m}$ for $s\le\frac{n}{2}$ and at most $e^{2(n-s)m\gamma}\binom{n}{s}^{-m}$ for $s\ge\frac{n}{2}$.
\end{lemma}
\begin{proof}
Without loss of generality, let $S=[s]$, i.e., $S$ is the first $s$ integers of $[n]$. 
Then for a permutation to not induce an edge between $S$ and $[n]\setminus S$, the permutation can be decomposed into a permutation of the first $s$ integers and a permutation of the remaining $n-s$ integers. 
Hence, there are $s!(n-s)!$ permutations of $[n]$ such that $S$ is preserved. 
Let $\Pi_S$ be the set of permutations that preserves $S$ so that $|\Pi_S|=s!(n-s)!$. 

For each permutation $\pi\in\Pi_S$, we define a subset $C_\pi$ of permutations so that (1) $C_{\pi'}\cap C_{\pi}=\emptyset$ for all $\pi,\pi'\in\Pi_S$ with $\pi\neq\pi'$, (2) $\pi$ is the only permutation of $C_\pi$ that preserves $S$, (3) $|C_\pi|=\binom{n}{s}$, and (4) $\pi$ and $\pi'$ have swap distance at most $2s$ for any $\pi'\in C_\pi$, hence implying that $\PPr{\calS=\pi}\le e^{2s\gamma}\cdot\PPr{\calS=\pi'}$. 
Recall that since $\pi\in\Pi_S$, then $\pi$ can be decomposed into permutations $\pi_1$ of the first $s$ integers and permutations $\pi_2$ of the remaining $n-s$ integers. 

Let $A$ be any set of $s$ indices of $[n]$, sorted in increasing order.  
Consider the following transformation $T_A$ on a permutation $\pi$ to produce a permutation $\psi$. 
Place the elements of $\pi$ in positions $[s]$ in order into the $s$ indices of $A$, so that $\pi'(A_i)=\pi(i)$. 
For the supplanted indices that have not been assigned to indices in $A$, place them in order into the remaining positions of $[s]$. 
Formally, let $X=[s]\setminus A$ and $Y=A\setminus[s]$. 
Then we set $\pi'(X_i)=\pi(Y_i)$ for all $i\in[|X|]$, noting that $|X|=|Y|$. 
We then define $C_\pi$ to be the set of permutations that can be obtained from this procedure, i.e., $C_\pi=\{\pi':\exists A\text{ with }\pi=T_A(\pi)\}$. 
See \figref{fig:swaps} for an example of the application of such an example $T_A$. 

\begin{figure*}[!htb]
\centering
\begin{tikzpicture}[scale=0.95]
\draw [decorate, very thick, decoration = {brace}] (0.8,6.5) --  (4.2,6.5);
\node at (2,7){$[s]$};

\node at (1,6){2};
\node at (2,6){4};
\node at (3,6){3};
\node at (4,6){1};
\node at (5,6){7};
\node at (6,6){5};
\node at (7,6){8};
\node at (8,6){6};
\node at (9,6){9};

\draw[thick] (1.7,5.5) -- (2.3,5.5);
\draw[thick] (2.7,5.5) -- (3.3,5.5);
\draw[thick] (5.7,5.5) -- (6.3,5.5);
\draw[thick] (7.7,5.5) -- (8.3,5.5);
\draw[<-] (2,5.3) -- (3.7,5);
\draw[<-] (3,5.3) -- (3.9,5);
\draw[<-] (6,5.3) -- (4.1,5);
\draw[<-] (8,5.3) -- (4.3,5);
\node at (4,4.7){$A$};

\node at (-2,4){(2, 4, 3, 1)};

\node at (1,4){?};
\node at (2,4){?};
\node at (3,4){?};
\node at (4,4){?};
\node at (5,4){7};
\node at (6,4){5};
\node at (7,4){8};
\node at (8,4){6};
\node at (9,4){9};

\draw[thick] (1.7,3.5) -- (2.3,3.5);
\draw[thick] (2.7,3.5) -- (3.3,3.5);
\draw[thick] (5.7,3.5) -- (6.3,3.5);
\draw[thick] (7.7,3.5) -- (8.3,3.5);

\node at (-2,3){(5, 6)};

\node at (1,3){?};
\node at (2,3){2};
\node at (3,3){4};
\node at (4,3){?};
\node at (5,3){7};
\node at (6,3){3};
\node at (7,3){8};
\node at (8,3){1};
\node at (9,3){9};

\draw[thick] (1.7,2.5) -- (2.3,2.5);
\draw[thick] (2.7,2.5) -- (3.3,2.5);
\draw[thick] (5.7,2.5) -- (6.3,2.5);
\draw[thick] (7.7,2.5) -- (8.3,2.5);

\node at (1,2){5};
\node at (2,2){2};
\node at (3,2){4};
\node at (4,2){6};
\node at (5,2){7};
\node at (6,2){3};
\node at (7,2){8};
\node at (8,2){1};
\node at (9,2){9};

\draw[thick] (1.7,1.5) -- (2.3,1.5);
\draw[thick] (2.7,1.5) -- (3.3,1.5);
\draw[thick] (5.7,1.5) -- (6.3,1.5);
\draw[thick] (7.7,1.5) -- (8.3,1.5);
\end{tikzpicture}
\caption{An example of the transformation $T_A$ for the permutation $\pi=(8,4,6,2,1,3,7,5,9)$, with $n=9$, $s=4$, and $A=(2,3,6,8)$. 
Note that the order $(8,4,6,2)$ is preserved within the indices of $A$ in the resulting permutation $\pi'=T_A(\pi)$ and the order $(3,5)$ is preserved within the indices $[s]-A$.  
}
\figlab{fig:swaps}
\end{figure*}
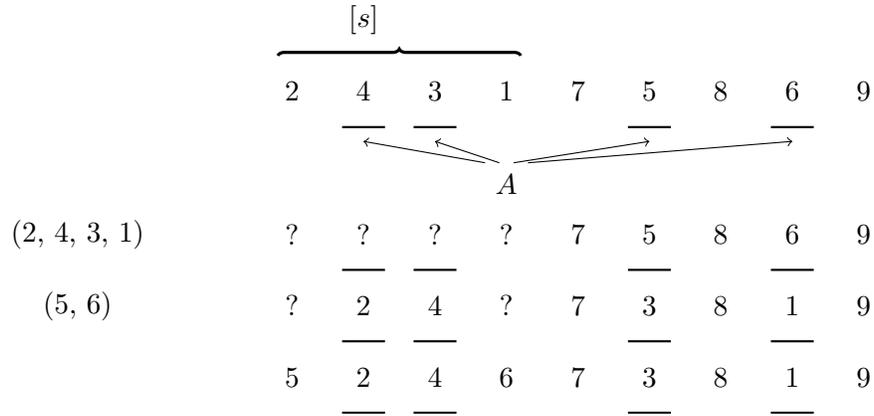

We first claim that $C_{\pi'}\cap C_{\pi}=\emptyset$ for all $\pi,\pi'\in\Pi_S$ with $\pi\neq\pi'$. 
Suppose by way of contradiction, there exists $\psi\in C_{\pi}\cap C_{\pi'}$, so that there exist sets $A$ and $A'$ with $\psi=T_A(\pi)=T_{A'}(\pi')$. 
Recall that since $\pi,\pi'\in\Pi_S$, then $\pi,\pi'$ can be decomposed into permutations $\pi_1,\pi'_1$ of the first $s$ integers and permutations $\pi_2,\pi'_2$ of the remaining $n-s$ integers. 
After applying $T_A$ to $\pi$, then the first $s$ integers are in the indices of $A'$, in some order. 
Similarly, after applying $T_{A'}$ to $\pi'$, then the first $s$ integers are in the indices of $A$, in some order. 
Hence for $\psi=T_A(\pi)=T_{A'}(\pi')$, it follows that $A=A'$, so it suffices to show that $T_A$ is injective for a fixed $A$. 

To that end, note that $T_A$ preserves the order of $[s]$ within $A$ and thus for $\pi=\pi_1\circ\pi_2$, then $\pi_1$ is the restriction of $T_A(\pi)$ to $A$. 
Similarly, note that $T_A$ does not touch the indices outside of $A\cup[s]$ and so $\pi_2$ is preserved by $T_A(\pi)$ in the restriction of $[n]\setminus(A\cup [s])$. 
Finally, $T_A$ preserves the relative order of $\pi_2$ inside the indices of $[s]\setminus A$. 
Therefore, given $A$ and $T_A(\pi)$, we can completely recover $\pi_1$ and $\pi_2$ and thus $\pi$. 
In other words, $T_A$ is injective, so that $T_A(\pi)=T_A(\pi')$ implies $\pi=\pi'$, which is a contradiction. 
\pasin{I don't understand the previous sentence... $T_{A}$ can be well-defined but many-to-one, right?}
\samson{Yes sorry, added sentences that $T_A$ is injective}
Hence $C_{\pi'}\cap C_{\pi}=\emptyset$. 

To see that $\pi$ is the only permutation of $C_\pi$ that preserves $S$, note that if any of the $s$ positions are picked outside $[s]$, then the resulting permutation places a value of $[s]$ outside of the first $s$ positions and so the resulting permutation does not preserve $S$, i.e., the values of $[s]$ are not retained within the first $s$ positions. 
However, there is only a single way to pick $s$ indices from $[n]$ that are all inside $[s]$, which corresponds to $\pi$. 
Hence, $\pi$ is the only permutation of $C_\pi$ that preserves $S$. 

To see the third property, note that $A$ is formed by choosing $s$ indices of $[n]$. 
Hence, $|A|=\binom{n}{s}$. 
Since $A$ is exactly the set of positions for which $\pi_1$ is mapped to, then each element of $A$ corresponds to a unique element in $C_\pi$. 
Thus, $|C_\pi|=\binom{n}{s}$. 

To see the fourth property, note that the only swaps are indices in $A$ with indices in $[s]$, meaning that at most $2s$ indices are changed. 
Thus we have $\pi$ and $\pi'$ have swap distance at most $2s$ for any $\pi'\in C_\pi$. 
Then by the $\gamma$-imperfect shuffle property, $\PPr{\calS=\pi}\le e^{2s\gamma}\cdot\PPr{\calS=\pi'}$. 

Since we have associated each $\pi\in\Pi_S$ with a set $C_\pi$ of size $\binom{n}{s}$ such that $\pi'\not\in C_\pi$ for $\pi'\in \Pi_S$ with $\pi'\neq\pi$ and $\PPr{\calS=\pi}\le e^{s\gamma}\cdot\PPr{\calS=\pi'}$, then it follows from a coupling argument that the probability that there exists no edge from $S$ to $[n]\setminus S$ after one iteration of the $\gamma$-imperfect shuffle is at most $e^{2s\gamma}\binom{n}{s}^{-1}$. 
By independence, the probability that there exists no edge from $S$ to $[n]\setminus S$ in $G$ after the $m$ iterations is at most $e^{2sm\gamma}\binom{n}{s}^{-m}$. 

By symmetry for sets $S$ with size $s$ and $n-s$, we have the probability is at most 
\[\min\left(e^{2sm\gamma}\binom{n}{s}^{-m},e^{2(n-s)m\gamma}\binom{n}{s}^{-m}\right)\]
across all ranges of $s$. 
\end{proof}

\begin{lemma}
\lemlab{lem:connect:one:prob:second}
Let $G$ be the communication graph of a $\gamma$-imperfect shuffler (on an $n$-player $m$-message protocol). 
For a fixed set $S$ with size $s$, the probability that there exists no edge from $S$ to $[n]\setminus S$ in $G$ is at most $e^{km\gamma}\binom{n/2}{k}^{-m}$ for any integer $k$ with $0\le k\le\min(s,n-s)$.
\end{lemma}
\begin{proof}
We can similarly show that the probability that there exists no edge from $S$ to $[n]\setminus S$ in $G$ after the $m$ iterations is at most $e^{km\gamma}\binom{n/2}{k}^{-m}$ for any integer $k$ with $0\le k\le\min(s,n-s)$ by the following modifications to the coupling argument. 
We again let $S=[s]$ without loss of generality and let $k \leq \min(s,n-s)$ be a fixed non-negative integer. 

Recall that there are $s!(n-s)!$ permutations of $[n]$ such that $S$ is preserved. 
We define $\Pi_S$ to be the set of permutations that preserves $S$ so that $|\Pi_S|=s!(n-s)!$ and we define a transformation $T_A(\pi)$ for a permutation $\pi\in\Pi_S$ as follows. 

If $s\le\frac{n}{2}$, we let $A$ be a set of $k$ positions in $\{s+1,\ldots,n\}$, sorted in increasing order. 
We then initialize $\psi=\pi$ and iteratively perform the following procedure $k$ times. 
For each $i\in[k]$, we swap the value in the $i$-th index of $\psi$ with the value in the $A_i$-th index of $A$. 
We then output set $T_A(\pi)$ to be the result of $\psi$ after applying these $k$ swaps. 
Note that since $[s]$ and $A$ are disjoint, we can also explicitly define the resulting $\psi=T_A(\pi)$ by
\[\psi(i)=\begin{cases}
&\pi(i),\qquad i\notin(A\cup[k])\\
&\pi(A_i),\qquad i\in[k]\\
&\pi(j),\qquad j=A_i, i\in[k].
\end{cases}.\]

Similarly, if $s\ge\frac{n}{2}$, we let $A$ be a set of $k$ positions in $[n-s]$, sorted in increasing order, and initialize $\psi=\pi$. 
Then for each $i\in[k]$, we swap the value in the $(n-i+1)$-th index of $\psi$ with the value in the $i$-th index of $A$. 
Alternatively, we can also explicitly define the resulting $\psi=T_A(\pi)$ by
\[\psi(i)=\begin{cases}
&\pi(i),\qquad i\notin(A\cup\{n-k+1,\ldots,n\})\\
&\pi(A_i),\qquad i\in\{n-k+1,\ldots,n\}\\
&\pi(j),\qquad j=A_i, i\in[k].
\end{cases}.\]

We again define $C_\pi$ to be the set of permutations that can be obtained from this procedure, i.e., $C_\pi=\{\pi':\exists A\text{ with }\pi=T_A(\pi)\}$. 
By the same argument as in \lemref{lem:connect:one:prob:second}, we have (1) $C_{\pi'}\cap C_{\pi}=\emptyset$ for all $\pi,\pi'\in\Pi_S$ with $\pi\neq\pi'$, (2) $\pi$ is the only permutation of $C_\pi$ that preserves $S$, (3) $|C_\pi|=\binom{n}{k}$. 
By the construction of $T_A$ performing $k$ swaps on $\pi$, we also have that $\pi$ and $\pi'$ have swap distance at most $k$ for any $\pi'\in C_\pi$, so that $\PPr{\calS=\pi}\le e^{k\gamma}\cdot\PPr{\calS=\pi'}$. 

Also by construction, we have $|C_\pi|\ge\binom{n/2}{k}$ and so by adapting the above coupling argument, we have that the probability that there exists no edge from $S$ to $[n]\setminus S$ in $G$ after the $m$ iterations is at most $e^{km\gamma}\binom{n/2}{k}^{-m}$. 
\end{proof}

By \lemref{lem:connect:one:prob:first} and \lemref{lem:connect:one:prob:second}, we have:
\begin{lemma}
\lemlab{lem:connect:one:prob}
Let $G$ be the communication graph of a $\gamma$-imperfect shuffler (on an $n$-player $m$-message protocol). 
For a fixed set $S$ with size $s$, the probability that there exists no edge from $S$ to $[n]\setminus S$ in $G$ is at most $e^{2sm\gamma}\binom{n}{s}^{-m}$ for $s\le\frac{n}{2}$, at most $e^{2(n-s)m\gamma}\binom{n}{s}^{-m}$ for $s\ge\frac{n}{2}$, and at most $e^{km\gamma}\binom{n/2}{k}^{-m}$ for any integer $k$ with $0\le k\le\min(s,n-s)$.
\end{lemma}

\lemref{lem:do:postprocess} and \lemref{lem:connect:one:prob} are the two main structural properties of imperfect shufflers that we use to overcome the challenge of adapting the analysis of \cite{BalleBGN20} to shufflers without symmetry. 

We now upper bound the probability that the number of connected components of $G$ is $c$, where $G$ is the underlying communication graph for the split-and-mix-protocol under a $\gamma$-imperfect shuffle. 
\begin{lemma}
\lemlab{lem:cc:prob}
Let $n\ge 19$ and $m\ge 8e^{4\gamma}$. 
Let $G$ be the communication graph of a $\gamma$-imperfect shuffler (on an $n$-player $m$-message protocol). 
Let $p(n,c)$ denote the probability that the number of connected components of $G$ is $c$. 
\pasin{One thing I am a little unsure is why we can decompose to the term $p(|S|,1)p(n-|S|,c-1)$, i.e. are the event that $S$ has a single component and $n - |S|$ has $c - 1$ components independent? This is not obvious to me since our shuffler is not uniform, so conditioning on the fact that $S$ is separated from $n - S$ does not seem to imply that these two events become independent. In any case, we can probably get way with this by just ignoring the term $p(|S|,1)$ completely anyway. Wdyt?}
\samson{Yes I agree. In fact, is it unclear whether $n$ refers to graph size or the subgraph size?}
Then 
\[p(n,c)\le\frac{2^{c-1}}{c!}\left(\frac{e}{n}\right)^{\frac{(m-1)(c-1)}{32e^{4\gamma}}}\cdot e^{2\gamma(m-1)(c-1)}.\]
\end{lemma}
\begin{proof}
For a fixed set $S$, let $\mathbb{P}_S$ denote the probability that there is no edge from $S$ to $[n]\setminus S$. 
Let $p(n,c)$ denote the probability that the number of connected components of $G$ is $c$. 
Then
\begin{align*}
p(n,c)&=\frac{1}{c}\sum_{S\subseteq[n]}\mathbb{P}_S\cdot p(n-|S|,c-1)\\
&\le\frac{1}{c}\sum_{s=1}^{n-c+1}\binom{n}{s}\mathbb{P}_S\cdot p(n-|S|,c-1).
\end{align*}
We decompose this sum and apply \lemref{lem:connect:one:prob}. 


By \lemref{lem:connect:one:prob}, we have $\mathbb{P}_S\le \min(e^{2(n-s)m\gamma}\binom{n}{s}^{-m},e^{2sm\gamma}\binom{n}{s}^{-m})$. 
By \lemref{lem:connect:one:prob}, we also have $\mathbb{P}_S\le e^{k m\gamma}\binom{n/2}{k}^{-m}$ for any $k\le\min(s,n-s)$. 
Observe that for $k\ge n-s\ge\frac{n}{2}$, we have $e^{2(n-s)m\gamma}\binom{n}{s}^{-m}\le e^{2km\gamma}\binom{n}{k}^{-m}\le e^{2km\gamma}\binom{n/2}{k}^{-m}$. 
Thus for $k=\frac{n}{4e^{4\gamma}}$, 
\begin{align*}
p(n,c)&\le\frac{1}{c}\sum_{s=1}^{k}\binom{n}{s}\binom{n}{s}^{-m}e^{2sm\gamma}\cdot p(n-|S|,c-1)\\
&+\frac{1}{c}\sum_{s=k+1}^{n-c+1}\binom{n}{s}\binom{n/2}{k}^{-m}e^{2km\gamma}\cdot p(n-|S|,c-1).
\end{align*}
Observe that $k=\frac{n}{4e^{4\gamma}}$ implies that
\begin{align*}
e^{2\gamma}&\le\left(\frac{n}{2k}\right)^{1/2}\\
e^{2km\gamma}&\le\left(\frac{n}{2k}\right)^{km/2}\le\binom{n/2}{k}^{m/2}\\
\binom{n/2}{k}^{-m}e^{2km\gamma}&\le\binom{n/2}{k}^{-m/2}\le\binom{n}{k}^{-m/2}.
\end{align*}
Thus we have
\begin{align*}
p(n,c)&\le\frac{1}{c}\sum_{s=1}^{k}\binom{n}{s}\binom{n}{s}^{-m}e^{2sm\gamma}\cdot p(n-|S|,c-1)\\
&+\frac{1}{c}\sum_{s=k+1}^{n-c+1}\binom{n}{s}\binom{n}{k}^{-m/2}\cdot p(n-|S|,c-1).
\end{align*}
Since $k=\frac{n}{4e^{4\gamma}}$, then
\begin{align*}
\binom{n}{k}^{-m/2}&\le(4e^{4\gamma})^{-\frac{nm}{8e^{4\gamma}}}\le (2e)^{-\frac{nm}{8e^{4\gamma}}}\le\binom{n}{n/2}^{-\frac{m}{4e^{4\gamma}}}\le\binom{n}{s}^{-\frac{m}{4e^{4\gamma}}}.
\end{align*}
Hence,
\begin{align*}
p(n,c)&\le\frac{1}{c}\sum_{s=1}^{k}\binom{n}{s}\binom{n}{s}^{-m}e^{2sm\gamma}\cdot p(n-|S|,c-1)\\
&+\frac{1}{c}\sum_{s=k+1}^{n-c+1}\binom{n}{s}^{1-\frac{m}{4e^{4\gamma}}}\cdot p(n-|S|,c-1).
\end{align*}
For $m\ge 8e^{4\gamma}$, we have $1\le\frac{m}{8e^{4\gamma}}$ and thus
\begin{align*}
p(n,c)&\le\frac{1}{c}\sum_{s=1}^{k}\binom{n}{s}\binom{n}{s}^{-m}e^{2sm\gamma}\cdot p(n-|S|,c-1)\\
&+\frac{1}{c}\sum_{s=k+1}^{n-c+1}\binom{n}{s}^{-\frac{m}{8e^{4\gamma}}}\cdot p(n-|S|,c-1).
\end{align*}
We first apply the induction hypothesis that $p(n,c)\le\frac{2^{c-1}}{c!}\left(\frac{e}{n}\right)^{\frac{(m-1)(c-1)}{32e^{4\gamma}}}\cdot e^{\gamma(m-1)(c-1)}$:
\begin{align*}
p(n,c)&\le\frac{2^{c-1}}{c!}\left(\frac{e}{n}\right)^{\frac{(m-1)(c-1)}{32e^{4\gamma}}}\cdot e^{2\gamma(m-1)(c-1)}\cdot\frac{1}{2}\cdot e^{\frac{(1-m)}{32e^{4\gamma}}}\cdot e^{2\gamma(1-m)}\\
&\cdot\left(\sum_{s=1}^{k}\binom{n}{s}^{1-m}e^{2sm\gamma}\left(\frac{n^{c-1}}{(n-s)^{c-2}}\right)^{\frac{m-1}{32e^{4\gamma}}}+\sum_{s=k+1}^{n-c+1}\left(\frac{(n-s)!s!n^{c-1}}{n!(n-s)^{c-2}}\right)^{\frac{m-1}{32e^{4\gamma}}}\right).
\end{align*}
We upper bound $p(n,c)$ by upper bounding the summation across the first $k$ terms, i.e., the head of the summation, then upper bounding the tail terms of the summation, i.e., the terms with $s\ge\frac{3n}{4}$, and finally upper bounding the remaining terms of the summation, i.e., $s\in\left[k,\frac{3n}{4}\right]$.

\paragraph{Upper bounding the head terms in the summation.} 
We now upper bound the summation across all $s\le k$. 
Let $a_s=\binom{n}{s}^{1-m}e^{2sm\gamma}\left(\frac{n^{c-1}}{(n-s)^{c-2}}\right)^{\frac{m-1}{32e^{4\gamma}}}$. 
For $s\le k=\frac{n}{4e^{4\gamma}}$ and $m\ge 8e^{4\gamma}$,
\begin{align*}
\frac{a_s}{a_{s-1}}&=\left(\frac{s}{n-s+1}\right)^{m-1}e^{2m\gamma}\left(\frac{n-s+1}{n-s}\right)^{\frac{(m-1)(c-2)}{32e^{4\gamma}}}\\
&\le\left(\frac{1}{8e^{4\gamma}}\right)^{m-1}e^{2m\gamma} e^{\frac{(m-1)(c-2)}{n-s}}\\
&\le\left(\frac{1}{8e^{4\gamma}}\right)^{m-1}(e^{4\gamma})^{m-1} e^{\frac{4(m-1)}{3}}\\
&\le\left(\frac{e^{4/3}}{8}\right)^{m-1}\le\left(\frac{1}{2}\right)^{m-1}\le\frac{1}{25}.
\end{align*}
Then through a geometric series, we bound the summation
\begin{align*}
\sum_{s=1}^k a_s&\le\sum_{s=1}^\infty\frac{a_1}{25^{s-1}}\le\frac{26a_1}{25}\\
&\le\frac{26}{25}n^{1-m} e^{m\gamma}\left(\frac{n^{c-1}}{(n-1)^{c-2}}\right)^{\frac{m-1}{32e^{4\gamma}}}\\
&\le\frac{26}{25}e^{m\gamma} e^{\frac{m-1}{32e^{4\gamma}}}
\end{align*}

\paragraph{Upper bounding the tail terms in the summation.}
We now upper bound the summation across all $s\ge\lceil\frac{3n}{4}\rceil$. 
Let $b_s=\left(\frac{(n-s)!s!n^{c-1}}{n!(n-s)^{c-2}}\right)^{\frac{m-1}{32e^{4\gamma}}}$. 
Then for $s\ge\frac{3n}{4}$,
\begin{align*}
\frac{b_s}{b_{s-1}}&=\left(\frac{s}{n-s+1}\left(\frac{n-s+1}{n-s}\right)^{c-2}\right)^{\frac{m-1}{32e^{4\gamma}}}\\
&\ge\left(\frac{s}{n-s}\right)^{\frac{m-1}{32e^{4\gamma}}}\ge 9.
\end{align*}
We again bound another subset of the sum through a geometric series:
\begin{align*}
\sum_{s=\ceil{3n/4}}^{n-c+1}b_s&\le\sum_{s=\ceil{3n/4}}^{n-c+1}\frac{b_{n-c+1}}{9^{n-c+1-s}}\\
&\sum_{s=-\infty}^{n-c+1}\frac{b_{n-c+1}}{9^{n-c+1-s}}\\
&=\frac{9b_{n-c+1}}{8}\\
&=\frac{9}{8}\left(\frac{(c-1)!(n-c+1)!n^{c-1}}{n!(c-1)^{c-2}}\right)^{\frac{m-1}{32e^{4\gamma}}}.
\end{align*}
Similar to \cite{BalleBGN20}, we bound the last expression using Sterling's bound, $\sqrt{2\pi}n^{n+\frac{1}{2}}e^{-n}\le n!\le en^{n+\frac{1}{2}}e^{-n}$, so that
\[\frac{9}{8}\left(\frac{(c-1)!(n-c+1)!n^{c-1}}{n!(c-1)^{c-2}}\right)^{\frac{m-1}{32e^{4\gamma}}}\le\frac{9}{8}\left(\frac{e}{\sqrt{2\pi}}(c-1)^{1.5}\left(1-\frac{(c-1)}{n}\right)^{n-c+1.5}\right)^{\frac{m-1}{32e^{4\gamma}}},\]
which is maximized at $c=3$ for $n\ge 19$, $m\ge 8e^{4\gamma}$, and $c\le\frac{n}{4}$. 
Thus,
\begin{align*}
\frac{9}{8}\left(\frac{e}{\sqrt{2\pi}}(c-1)^{1.5}\left(1-\frac{(c-1)}{n}\right)^{n-c+1.5}\right)^{\frac{m-1}{32e^{4\gamma}}}&\le\frac{9}{8}\left(\frac{2e}{\sqrt{\pi}}\left(1-\frac{2}{n}\right)^{n-1.5}\right)^{\frac{m-1}{32e^{4\gamma}}}\\
&\le\frac{9}{8}(1.27)^{\frac{m-1}{32e^{4\gamma}}}.
\end{align*}

\paragraph{Upper bounding the middle terms in the summation.} 
It remains to upper bound the summation across $s\in\left[\frac{n}{4e^{4\gamma}},\frac{3n}{4}\right]$. 
We have for $\alpha=\frac{s}{n}$,
\[b_s=\left(\frac{((1-\alpha)n)!(\alpha n)!}{(n-1)!(1-\alpha)^{c-2}}\right)^{\frac{m-1}{32e^{4\gamma}}}.\]
By Sterling's bound, we have
\[b_s\le\left(\frac{e^2}{\sqrt{2\pi}}\sqrt{n}(1-\alpha)^{2.5-c+(1-\alpha)n}\alpha^{\alpha n+\frac{1}{2}}\right)^{\frac{m-1}{32e^{4\gamma}}}\le\left(\frac{e^2\sqrt{n}}{\sqrt{2\pi}}\alpha^{\alpha n}\right)^{\frac{m-1}{32e^{4\gamma}}}.\]
Since there are at most $n$ such terms $b_s$, then
\[\sum_{s=k+1}^{\lceil 3n/4\rceil-1} b_s\le n\left(\frac{e^2\sqrt{n}}{\sqrt{2\pi}}\left(\frac{3}{4}\right)^{\frac{3n}{4}}\right)^{\frac{m-1}{32e^{4\gamma}}}\le2\left(en\left(\frac{3}{4}\right)^{\frac{3n}{4}}\right)^{\frac{m-1}{32e^{4\gamma}}}\le 2.\]
\paragraph{Putting things together.}
Combining the upper bounds across the three summations, we have
\begin{align*}
\sum_{s=1}^{k}\binom{n}{s}^{1-m}e^{2sm\gamma}&\left(\frac{n^{c-1}}{(n-s)^{c-2}}\right)^{\frac{m-1}{32e^{4\gamma}}}+\sum_{s=k+1}^{n-c+1}\left(\frac{(n-s)!s!n^{c-1}}{n!(n-s)^{c-2}}\right)^{\frac{m-1}{32e^{4\gamma}}}\\
&\le\frac{26}{25}e^{m\gamma} e^{\frac{m-1}{32e^{4\gamma}}}+2+\frac{9}{8}(1.27)^{\frac{m-1}{32e^{4\gamma}}}\\
&\le 2e^{\frac{m-1}{32e^{4\gamma}}}\cdot e^{m\gamma}\le 2e^{\frac{m-1}{32e^{4\gamma}}}\cdot e^{2\gamma(m-1)}.
\end{align*}
Therefore, we have
\[p(n,c)\le\frac{2^{c-1}}{c!}\left(\frac{e}{n}\right)^{\frac{(m-1)(c-1)}{32e^{4\gamma}}}\cdot e^{2\gamma(m-1)(c-1)},\]
as desired.
\end{proof}

We now upper bound the expected value of $\Ex{q^{C(G)}}$ for the purposes of upper bounding the right hand side of \lemref{lem:prob:shuf:exp:graph}. 
\begin{lemma}
\lemlab{lem:exp:qconn:bound}
Let $n\ge 19$, $m\ge 8e^{4\gamma}$, and $q\le\left(\frac{n}{e}\right)^{\frac{(m-1)}{32e^{4\gamma}}}e^{2\gamma(1-m)}$. 
Let $G$ be the graph on $n$ vertices formed a random instantiation of the split-and-mix protocol $\calP_{m,n}$ with $m$ messages for each of $n$ players, using a $\gamma$-imperfect shuffler $\calS$. 
That is, let $G$ have an edge between $i$ and $j$ if and only if player $i$ passes one of their $m$ messages to player $j$. 
Then 
\[\Ex{q^{C(G)}}\le q+3q^2 e^{2\gamma(m-1)}\left(\frac{e}{n}\right)^{\frac{m-1}{32e^{4\gamma}}}.\]
\end{lemma}
\begin{proof}
By \lemref{lem:cc:prob}, we have
\[p(n,c)\le\frac{2^{c-1}}{c!}\left(\frac{e}{n}\right)^{\frac{(m-1)(c-1)}{32e^{4\gamma}}}\cdot e^{2\gamma(m-1)(c-1)}.\]
Taking the expectation, we have
\[\Ex{q^{C(G)}}\le\sum_{c=1}^nq^c\frac{2^{c-1}}{c!}\left(\frac{e}{n}\right)^{\frac{(m-1)(c-1)}{32e^{4\gamma}}}\cdot e^{2\gamma(m-1)(c-1)}.\]
Since term in the summand after the second term is at most $\frac{2q}{3}\left(\frac{e}{n}\right)^{\frac{(m-1)}{32e^{4\gamma}}}e^{2\gamma(m-1)}$ times the previous term in the summand, then
\[\Ex{q^{C(G)}}\le q+q^2 e^{2\gamma(m-1)}\left(\frac{e}{n}\right)^{\frac{m-1}{32e^{4\gamma}}}\sum_{i=0}^\infty\left(\frac{2q}{3}\left(\frac{e}{n}\right)^{\frac{(m-1)}{32e^{4\gamma}}}e^{2\gamma(m-1)}\right)^i.\]
Since $q\le\left(\frac{n}{e}\right)^{\frac{(m-1)}{32e^{4\gamma}}}e^{2\gamma(1-m)}$ by assumption, then
\[\Ex{q^{C(G)}}\le q+3q^2 e^{2\gamma(m-1)}\left(\frac{e}{n}\right)^{\frac{m-1}{32e^{4\gamma}}}.\]
\end{proof}
We now analyze the statistical security of the split-and-mix protocol. 
\begin{lemma}
\lemlab{lem:stat:sec:ikos}
Let $n\ge 19$, $m\ge 8e^{4\gamma}$, and $q\le\left(\frac{n}{e}\right)^{\frac{(m-1)}{32e^{4\gamma}}}e^{2\gamma(1-m)}$. 
Then we have worst-case statistical security with parameter
\[\sigma\le(m-1)\left(\frac{\log n-\log e}{64e^{4\gamma}}-2\gamma\log e\right)-3\log(3q),\]
\end{lemma}
\begin{proof}
By \lemref{lem:exp:tvd} and \lemref{lem:prob:shuf:exp:graph}, we have
\[\mathbb{E}_{\vec{\sfX},\vec{\sfX}'}[\TVD(\calP_{m,n}(\vec{\sfX}),\calP_{m,n}(\vec{\sfX}'))]\le\sqrt{q^{mn-1}\Ex{q^{C(G)-mn}}-1},\]
where $C(G)$ is the communication graph for the shuffle $\calS^{-1}\circ\calS'$. 
By \lemref{lem:do:postprocess} and the fact that $\calS'$ is a $\gamma$-imperfect shuffler, we have that $\calS^{-1}\circ\calS'$ is also a $\gamma$-imperfect shuffler and thus it suffices to upper bound $\Ex{q^{C(G)-mn}}$ where $C(G)$ is the communication graph for an arbitrary $\gamma$-imperfect shuffler $\calS$. 
Therefore by \lemref{lem:exp:qconn:bound}, we have average case statistical security less than or equal to
\[2^{-\sigma}\ge\sqrt{3q^3e^{2\gamma(m-1)}\left(\frac{e}{n}\right)^{\frac{m-1}{32e^{4\gamma}}}},\]
which holds for 
\[\sigma\le(m-1)\left(\frac{\log n-\log e}{64e^{4\gamma}}-2\gamma\log e\right)-3\log(3q).\]
The claim then follows by the reduction of worst-case input to average-case input by \lemref{lem:avg:worst}. 
\end{proof}
Now it can be verified that by restricting $\gamma\le\frac{\log\log n}{80}$, then we have both $728e^{4\gamma}\le\log n$ and  $\ceil{2n^{3/2}}\le\left(\frac{n}{e}\right)^{\frac{(m-1)}{32e^{4\gamma}}}e^{2\gamma(1-m)}$. 
These conditions imply that 1) $\left(\frac{\log n-\log e}{64e^{4\gamma}}-2\gamma\log e\right)=\O{\frac{\log n}{e^{4y}}}$, so that the parameter $\sigma$ has a non-empty range in the statement of \lemref{lem:stat:sec:ikos}, and 2) $q=\ceil{2n^{3/2}}$ satisfies $q\le\left(\frac{n}{e}\right)^{\frac{(m-1)}{32e^{4\gamma}}}e^{2\gamma(1-m)}$ in the statement of \lemref{lem:stat:sec:ikos}. 
As a corollary, we obtain the following guarantees for worst-case statistical security:
\thmsecurecc*

\section{Conclusion and Discussion}
In this work, we introduce the imperfect shuffle DP model, as a means of abstracting out real-world scenarios that prevent perfect shuffling. 
We also give a real summation protocol with nearly optimal error and small communication complexity. 
The protocol, which is based on the split-and-mix protocol~\cite{IshaiKOS06}, is similar to that of the (perfect) shuffle model~\cite{BalleBGN20,GhaziMPV20}, while the main challenge comes in the analysis. 
Although we overcome this hurdle for this particular protocol, our techniques are quite specific. 
Therefore, an interesting open question is whether there is a general theorem that transfer the privacy guarantee in the perfect shuffle model to that in the imperfect shuffle model, possibly with some loss in the privacy parameters. 

Another interesting direction is whether a sub-exponential dependency on $\gamma$ is possible in the number of messages, as the current dependency in \thmref{thm:main} is $\O{e^{4\gamma}}$. 
It also remains an interesting open question whether our results can be extended to the setting where all the $mn$ messages are shuffled together using a single shuffler. 

Finally, it would be natural to consider an imperfect shuffler that handles additive error in the manner of approximate differential privacy. 
That is, what can we say about a variant of \defref{def:imp:shuffle} that permits an additive $\delta$ term along the lines of \defref{def:dp}?

\def\shortbib{0}
\bibliographystyle{alpha}
\bibliography{references}
\appendix
\section{Additional Proofs}
\applab{app:extra:proofs}
\lemsecuretodp*
\begin{proof}
Given $\Xi$, we construct a protocol $\calP$ on a field of size $q:=\ceil{2n^{3/2}}$ as follows. 
Each input $x_i\in[0,1]$ is rounded to a value $y_i$ precision $p=\sqrt{n}$, so that $y_i=\lfloor x_ip\rfloor+\Ber(x_ip-\lfloor x_ip\rfloor)$. 
We then add Polya noise to each term, so that $z_i=y_i+\Polya\left(\frac{1}{n},e^{-\eps/p}\right)-\Polya\left(\frac{1}{n},e^{-\eps/p}\right)$. 
The protocol $\calP$ then runs $\Xi$ on the inputs $z_1,\ldots,z_n$ to achieve a sum $Z$. 
The protocol $\calP$ then decodes $Y$ by outputting $\widetilde{X}=\frac{Z}{p}$ if $Z\le\frac{3np}{2}$ and by outputting $\widetilde{X}=\frac{Z-q}{p}$ otherwise if $Z>\frac{3np}{2}$. 

\paragraph{Upper bounding the expected error.}
There are multiple sources of error. 
The first source of error comes from the randomized rounding that produces the values $y_1,\ldots,y_n$ from $x_1,\ldots,x_n$. 
By \lemref{lem:round:err}, we have that
\[\Ex{\left(\sum_{i=1}^n\left(x_i-\frac{y_i}{p}\right)\right)^2}\le\frac{n}{4p^2}.\]
Since $p=\sqrt{n}$, then by Markov's inequality, 
\[\PPr{\left(\sum_{i=1}^n\left(x_i-\frac{y_i}{p}\right)\right)^2\ge\frac{n^2}{16}}\le\frac{4}{n^2}.\]

The second source of error comes from the noise added to the variables $y_1,\ldots,y_n$ to obtain the values $z_1,\ldots,z_n$. 
Beyond the error incurred from the randomized rounding, $\widetilde{X}$ has additional error distributed according to the discrete Laplacian distribution unless the total noise added has magnitude larger than $\frac{n}{2}$, which could potentially cause an additive $\O{n^2}$ squared error due to incorrect decoding of $\widetilde{X}$ from $Z$. 
By \factref{fact:polya:dlap}, the total noise added to the variables $y_i$ is a random variable drawn from the discrete Laplace distribution, i.e.,
\[\sum_{i=1}^n(z_i-y_i)\sim\DLap\left(e^{-\eps/p}\right).\]
By the distribution of the discrete Laplacian, the noise added by the discrete Laplacian distribution is more than $\frac{n}{4}$ with probability $\O{e^{-\eps n/4}}$. 
Therefore, the expected mean squared error in this protocol is at most $\O{\frac{1}{\eps^2}}+\frac{4}{n^2}\cdot\O{n^2}+\O{e^{-\eps n/4}}\cdot\O{n^2}=\O{\frac{1}{\eps^2}}$. 
By Jensen's inequality, the expected absolute error in this protocol is at most $\O{\frac{1}{\eps}}$. 

\paragraph{Privacy considerations.}
Consider the protocol $\Xi'$ with input $w_1=Z$ and $w_2=\ldots=w_n=0$, where $Z=z_1+\ldots+z_n$, as defined above. 
Note that by \factref{fact:polya:dlap}, we have $w_1=Z\sim\sum_{y=1}^n y_i+\DLap\left(e^{-\eps/p}\right)$. 
Since the sensitivity of $\sum_{y=1}^n y_i$ is $p$, then $Z$ is $(\eps,0)$-differentially private. 
Therefore, by post-processing, $\Xi'$ is also $(\eps,0)$-differentially private. 

Moreover, we can upper bound the statistical distance between $\Xi$ and $\Xi'$ by $2^{-\sigma}$ by a coupling argument. 
Specifically, by coupling the noise added to $x_i$ by both $\Xi$ and $\Xi'$, 
the protocols $\Xi$ and $\Xi'$ will output the same sum given this coupled randomness outside of the  $2^{-\sigma}$ difference induced by the $\sigma$-secure property. 
\end{proof}

%
%
%
%
\end{document}